\documentclass[a4paper,11pt]{amsart}
\usepackage{amsfonts,amssymb,epsfig,latexsym}
\usepackage{amsmath}
\usepackage[latin1]{inputenc}
\usepackage{color,layout}
\usepackage{ifthen}

\usepackage{dsfont}

\usepackage{pdfsync}
\usepackage{graphicx}
\usepackage{color}
\usepackage{pdfsync}
\date{}

\newcommand{\W}{{\mathcal W}}

\newcommand{\ra}{\text{Ran}}

\newcommand{\be}{\begin{equation}}
\newcommand{\ee}{\end{equation}}

\def\la{\langle}
\def\ra{\rangle}

\def\R{\mathbb{R}}
\def\C{\mathbb{C}}

\def\N{\mathbb{N}}

\renewcommand{\Re}{\text{{\rm Re}\;}}
\renewcommand{\Im}{\text{{\rm Im}\;}}

\newcommand{\re}{{\rm Re}}
\newcommand{\im}{{\rm Im}}
\newcommand{\ord}{{\mathcal O}}
\newcommand{\ai}{{\rm Ai}\,}
\newcommand{\bi}{{\rm Bi}\,}

\newtheorem{theorem}{Theorem}[section]
\newtheorem{lemma}[theorem]{Lemma}
\newtheorem{proposition}[theorem]{Proposition}
\newtheorem{corollary}[theorem]{Corollary}

\theoremstyle{definition}

\newtheorem{remark}[theorem]{Remark}

\numberwithin{equation}{section}

\title[Molecular dynamics]{Molecular dynamics at an energy-level crossing}
\author[Philippe BRIET \& 
Andr\'e MARTINEZ]{
Philippe BRIET${}^1$ \& 
Andr\'e MARTINEZ$ {}^2$
 }

\begin{document}

\maketitle 
\addtocounter{footnote}{1}
\footnotetext{{\tt\small  Aix-Marseille Universit\'{e}, CNRS, CPT UMR 7332, 13288 Marseille, France, and Universit\'{e} de Toulon, CNRS, CPT UMR 7332, 83957 La Garde, France,  
briet@univ-tln.fr} }  
\addtocounter{footnote}{1}
\footnotetext{{\tt\small Universit\`a di Bologna,  
Dipartimento di Matematica, Piazza di Porta San Donato, 40127
Bologna, Italy, 
andre.martinez@unibo.it }}  
\begin{abstract}
This paper   is a continuation of   a previous work \cite{BrMa} about the study  of the survival probability modelizing   the molecular predissociation   in the Born-Oppenheimer framework.   Here we consider the critical case where the reference energy corresponds to   the value of a crossing of two electronic levels, one of these two levels being confining while the second dissociates.
We show that   the survival probability   associated to a certain  initial  state   is a sum of the  usual time-dependent exponential contribution,  and a reminder term that is jointly   polynomially small with respect to the time and the semiclassical parameter. We also compute explicitly the main contribution of the remainder.
\end{abstract}  
\vskip 4cm
{\it Keywords:} Resonances; Born-Oppenheimer approximation; eigenvalue crossing; quantum evolution; survival probability.
\vskip 0.5cm
{\it Subject classifications:} 35P15; 35C20; 35S99; 47A75.

\baselineskip = 18pt 
\vfill\eject
\section{Introduction}

This paper  concerns  the   study  of the behaviour in time of  some  quantum states  describing  the predissociation process of  a molecular systems in the  Born Oppenheimer approximation. 
 Recall that  in this context, the  predissociation is connected  with  a resonant state of the system coming from an internal conversion from an excited state  towards a dissociative state when   the Born-Oppenheimer parameter   $h$ is small.  We refer  to a recent paper by the same authors  \cite{BrMa}   and references therein    for  more details.  

Here we consider the critical case where the reference energy $E=0$   corresponds to a crossing  of   the confining electronic  energy  curve  and  the  dissociative one. We suppose that the system  has  only one such crossing  point.   Despite the  absence of tunnelling  for $E$,  resonances    exist  \cite {FMW1}. They are of the form  $ \rho (h)= \lambda(h) + O(h^{\frac{4}{3}})$ where $\lambda(h)$ is   an eigenvalue  (embedded in the continuous  spectrum)   near $0$ of  the  decoupled operator, and their widths  satisfy   $ \Im  \rho (h) = O(h^{\frac{5}{3}})$ as $ h \sim 0$ (actually, under some assumption of non degeneracy of the coupling operator, one also know that $\im\rho(h) <0$: see \cite{FMW1, FMW2}).   Therefore,  an attention must be paid  to the dynamics of certain states having an energy close to that of the resonance.   

 As in the case studied in \cite {BrMa},   the initial state  $\phi $  is  the normalized  eigenvector   associated  with a simple eigenvalue   $\lambda(h)$ of  the  decoupled operator. Then, we show that  for $h$ small enough, $g$ a cut-off function supported near $\lambda(h)$, and $t\in \mathbb R^+$, the survival probability satisfies,
 \be
\label{1} {\mathcal A}_\phi  =  ( e^{-itH} g(H)\phi, \phi) =  e^{-it\rho(h)}b(\phi,h) + r(t, \phi, h),
\ee
where $ b(\phi,h)= 1 + {\mathcal O} (h^{\frac{1}{3}})$  and $  r(t,\phi, h) =   h^{\frac23} {\mathcal O}(\la ht\ra^{-\infty})$ (here we use the notation 
$\la s\ra := (1+ s^2)^{\frac{1}{2}}$).
We actually prove this   result  in a  situation where the  inter-level coupling  is  a general   first-order differential operator.  In the physical  model  the coupling operator is a vector-field (see \cite{FMW2}), and
we then  expect   a higher order estimate on the long time  part of  $ {\mathcal A}_\phi $  i.e. $r(t, \phi, h) = h^{\frac43} {\mathcal O}(\la ht\ra^{-\infty})$. This fact will be proved in a forthcoming paper  \cite{BrMa3}.

In contrast with previous papers on similar estimates (see, e.g., \cite{CGH, CoSo, Her, Hu2, JeNe}), here we also focus on the precise behaviour of the remainder term $r(t,\phi,h)$. We prove,
$$
r(t,\phi,h)=\alpha h^{\frac23}e^{-it\lambda(h)}F(ht) + \ord(h\la ht\ra^{-\infty}),
$$
where $\alpha$ behaves like a constant, and $F$ is an explicit analytic function on $\R^+$ (depending on $g$) that satisfies $F(0)\not=0$, $F(\lambda)=\ord(\la \lambda\ra^{-\infty})$ (see Theorem \ref{mainth} for the precise statement).

 In view of \eqref{1}, it turns out that the critical time $ t_{c}$,  within which the  contribution of the  exponential part of $ {\mathcal A}_\phi $   is preponderant  with respect to the  remainder term, satisfies,
 $$ t_{c}\geq  \frac{2}{3} \frac{|\ln(h)|}{\vert  \Im \rho\vert }.$$
(Recall that $\im\rho (h) = \ord( h^{\frac53})$.) This means that for time   $t \leq t_{c} $ the  strong  resonance effects  persist, while they  disappear for larger  times.   (Note that  for  the physical model,   we have $\Im \rho (h) =\ord( h^{\frac73})$ \cite{FMW2}, and we can  expect   that   $ t_{c}\geq  \frac{4}{3} \frac{|\ln(h)|}{\vert  \Im \rho\vert }$.)

  Concerning the proof, in addition to the techniques introduced in \cite{FMW1} we also use some special kinds of semiclassical function spaces that permit us to considerably facilitate the estimates on the remainder term $r(t,\phi, h)$.

 Let us describe the content of the paper. In  section 2 we give the assumptions and the   main result. The strategy of the proof  involving the distortion  theory will be described in section 3. Section 4 and 5 are devoted to  obtain convenient estimates on the resolvent operators. In the section 6, 7, 8 and 9 we  prove   estimates  on  the remainder term  in   the r.h.s of \eqref{1}.  The coefficient $b(h)$  is  studied in section 9.

\section{Assumptions and main result}
We consider the semiclassical $2\times 2$-matrix Schr\"odinger operator,
 $$
   H= \left(
\begin{matrix}
 P_1 & hW\\
hW^* & P_2
\end{matrix}
\right)\quad ; \quad P_j=h^2D_x^2 + V_j(x)
   $$
 where, as in \cite{FMW1, FMW2}, we assume,

{\bf Assumption (A1)}
$V_1(x)$, $V_2(x)$ 
 are real-analytic on $\R$ and extend to holomorphic functions
 in the complex domain,
$$
\Gamma=\{x\in\C;|\im x|<\varepsilon_0\la\re \,x\ra\}\quad ; \quad \la\re \,x\ra:= (1+|\re x|^2)^{\frac12},
$$
where $\varepsilon_0>0$ is a constant.

{\bf Assumption (A2)}For $j=1,2$, $V_j$ admits limits as $\re\, x\to \pm\infty$ in $\Gamma$, and they satisfy,
$$
\begin{aligned}
\lim_{{\re\,x\to -\infty}\atop{x\in \Gamma}}  V_1(x)>0\, ;\, \lim_{{\re\,x\to -\infty}\atop{x\in \Gamma}} V_2(x)>0\, ;\\
\lim_{{\re\,x\to +\infty}\atop{x\in \Gamma}} V_1(x)>0\, ;\, \lim_{{\re\,x\to +\infty}\atop{x\in \Gamma}} V_2(x)<0.
\end{aligned}
$$

{\bf Assumption (A3)} One has,
$$
V_1'(x^*)=:-\tau_0  <0,\quad  V_1'(0)=:\tau_1>0,\qquad V_2'(0)=:-\tau_2<0,
$$
and there exists a negative number $x^*<0$ such that,
\begin{itemize}
\item $V_1>0$ and $V_2>0$ on $(-\infty, x^*)$;
\item $V_1<0<V_2$ on $(x^*,0)$;
\item $V_2<0<V_1$ on $(0,+\infty)$.
\end{itemize}

{\bf Assumption (A4)}
$W(x,hD_x)$ is a first order differential operator
$$
W(x,hD_x)=a_0(x)+ia_1(x)hD_x,
$$
where $a_0(x)$ and $a_1(x)$ are analytic and bounded in $\Gamma$, and  real for real $x$.

\begin{figure}[h]
\label{fig1}
\begin{center}
\scalebox{0.5}[0.35]{
\includegraphics{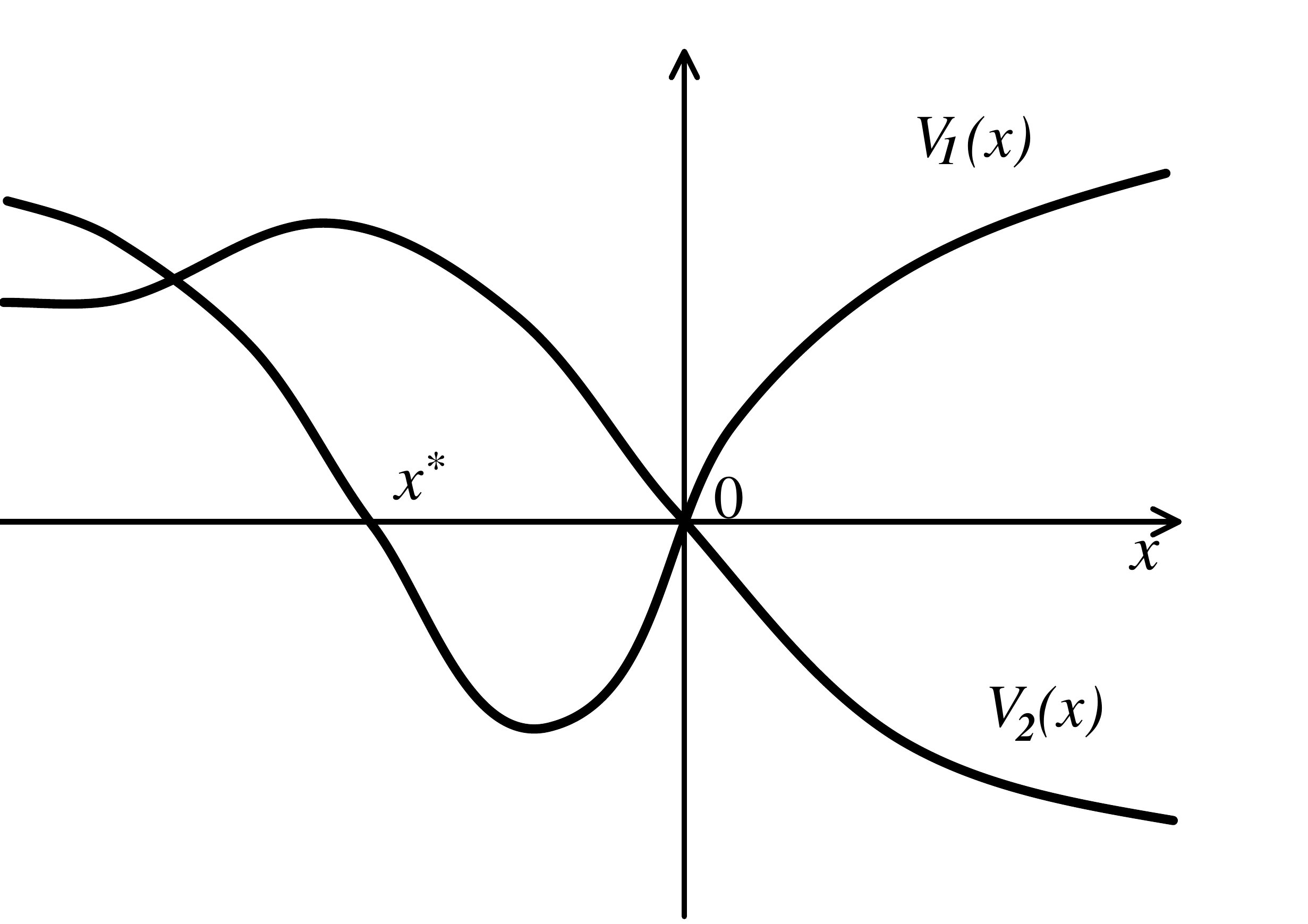}
}
\end{center}
\caption{The two potentials}
\end{figure}

In this situation, we know from \cite{FMW1} that the resonances of $H$ that are inside ${\mathcal D}_h(C_0):= [-C_0h^{2/3}, C_0h^{2/3}]-i[0,C_0h]$ ($C_0>0$ arbitrary) are of the form,
\be
\label{rhok}
\rho_k(h) = e_k(h) + \ord (h^{\frac43})\quad ; \quad \im\, \rho_k(h) = \ord (h^\frac53),
\ee
with $k\in \N$ and
$$
e_k(h):=\frac{-2{\mathcal A}(0)+(2k+1)\pi h}{2{\mathcal A}'(0)}\,\, ;\,\, {\mathcal A}(E):= \int_{x_1^*(E)}^{x_1(E)}\sqrt{ E-V_1(t)} \, dt,
$$
where $x_1^*(E)$ (respectively $x_1(E)$) is the unique solution of $V_1(x)=E$ close to $x^*$ (respectively close to 0). In addition, at each such $e_k(h)$ inside $[-C_0h^{\frac23}, C_0h^{\frac23}]$, corresponds a unique resonance $\rho_k(h)$ of $H$ that satisfies \eqref{rhok}. On the other hand, it is also well known (see, e.g., \cite{HeRo}) that at each such $e_k(h)$ inside $[-C_0h^{\frac23}, C_0h^{\frac23}]$, corresponds a unique eigenvalue $E_k(h)$ of $P_1$, such that,
\be
\label{rhok}
E_k(h) = e_k(h) + \ord (h^2),
\ee
From now on, we fix such an eigenvalue, that is, we choose once for all an application,
$$
h\mapsto \lambda_0(h)\in {\rm Sp}(P_1)\cap [-C_0h^{\frac23}, C_0h^{\frac23}],
$$
to which corresponds a unique application,
$$
h\mapsto \rho_{0} (h)\in {\rm Res}(H)\cap {\mathcal D}_h(C_0),
$$
such that,
$$
\rho_{0} (h) -\lambda_0(h) =\ord (h^{\frac43}).
$$
We also denote by $\varphi_0$ the real-valued normalized eigenfunction of $P_1$ associated with $\lambda_0$ (so that $W\varphi_0$ and $W^*\varphi_0$ are real-valued, too), and we set,
$$
\phi := (\varphi_0,0)\in L^2(\R)\oplus L^2(\R).
$$
In particular, there exists some complex number $c_0=c_0(h)\sim 1$ such that, for $x\leq 0$,
$$
\varphi_0=c_0h^{-\frac16}u_{1,L}^-(\lambda_0).
$$
 (Actually, by computing the $L^2$-norm of $u_{1,L}^-(\lambda_0)$ on $I_L$, one can see that $c_0^2 = \frac2{\pi} \int_{x^*}^0 \frac{dx}{\sqrt{\lambda_0-V_1(x)}}+\ord(h^\frac13)$.)

We also fix some cutoff function $g_0 \in C_0^\infty ((-\delta_1,\delta_1); [0,1])$ such that $g=1$ on $[-\delta_0,\delta_0]$ with $0<\delta_0<\delta_1 <\frac{\pi }{{\mathcal A}'(0)}$, so that, if we set,
\be
\label{defg}
g(\lambda) := g_0\left( \frac{\lambda -\lambda_0}{h}\right),
\ee
then, for $h$ small enough,  $\lambda_0$ is the only eigenvalue of $P_1$ contained in the support of $g$.

We are interested in the survival amplitude associated with $g(H)^{\frac12}\phi$,
$$
{\mathcal A}_{\phi }:= \la e^{-itH}g(H)\phi, \phi\ra.
$$
In order to state our result, we define,
\be
\label{defF}
F(\lambda):= -2i\int_{\gamma_0}\frac{e^{-i\lambda z}g_0(\re z)}{z^2}dz,
\ee
where $\gamma_0$  is the oriented complex path,
$$
\gamma_0 := (-\infty, -\delta_0]\cup \{ \delta_0e^{i\alpha}\, ;\, \alpha \in [\pi, 2\pi]\} \cup [\delta_0, +\infty).
$$
In particular, let us observe that $F$ is analytic, and that $F(0)\not=0$ (indeed, one can compute $F(0)=4i \alpha\delta_1^{-1}$ with $\alpha \geq 1$).
In addition, by integration by parts, we also see that $F(\lambda) =\ord (|\lambda|^{-\infty})$ as $\lambda\to \pm\infty$.

In the sequel, we denote by $\ai$ and $\bi$ the standard Airy functions, and for any function $f=f(s)$ we set $\check f(s) :=f(-s)$

Our main result is,
 \begin{theorem}\sl
\label{mainth}
Under assumptions {\bf (A1)-(A4)}, one has,
\be
\label{casnonfis}
 {\mathcal A}_\phi  =e^{-it\rho_{0}}b(h) + h^{\frac23}q_0(t,h) + {\mathcal O}(h\la ht\ra^{-\infty})
\ee
uniformly for $h>0$ small enough and $t\in \R$,
with, 
\be
\label{estbh}
b(h) =1+{\mathcal O}(h^{1/3});
\ee
\be
\label{estq0}
q_0(t,h)=4a_0(0)^2c_0^2e^{-it\lambda_0}\left[A_0(\lambda_0 h^{-\frac23})\right]^2F(ht),
\ee
where $F$ is defined in \eqref{defF}, and $A_0$  is the function,
$$
\begin{aligned}
A_0(s) :=  \tau_1^{-\frac16}\tau_2^{-\frac16}(\tau_1+\tau_2)^{-\frac13} \check\ai \left(\left(\frac{\tau_1+\tau_2}{\tau_1\tau_2} \right)^{\frac23}s \right).
\end{aligned}
$$
\end{theorem}

\section{Preliminaries}

As in \cite[Section 5]{BrMa}, we have,
\be
\label{stone2}
{\mathcal A}_{\phi } =  e^{-it\rho_{0}}b(\phi,  h)+ r(t,\phi, h),
\ee
where $b(\phi,  h)$ is the residue at $\rho_{0}$ of the meromorphic function 
$$
z\mapsto - \la R_\theta (z)\phi_\theta, \phi_{-\theta}\ra,
$$
where, $R_\theta (z):=U_\theta (H_\theta -z)^{-1}U_\theta^{-1}$ is the distorted resolvent of $H$, and  $\phi_\theta =(\varphi_0^\theta, 0):=U_{\theta}\phi$ is the analytic distortion of $\phi$, where $U_\theta$ is the analytic distorsion given by,
$$
U_{\theta}\phi (x) = \phi (x+i\theta \nu (x)),
$$
with $\nu \in C^\infty (\R; \R)$, $\nu  =0$ on $(-\infty, x_\infty]$ for some $x_\infty >0$, $\nu (x)=x$ for $x$ large enough (see \cite[Section 3]{FMW1}).

Further, $r(t,\phi, h)$ is given by,
\be
\label{reste}
r(t,\phi, h):= \frac1{2i\pi}\int_{\gamma_-} e^{-itz} g(\Re z) \left(\la R_\theta (z)\phi_\theta, \phi_{-\theta}\ra -\la R_{-\theta}(z)\phi_{-\theta}, \phi_\theta\ra\right) dz,
\ee
where $\gamma_-$ is a complex contour parametrized by $\Re z$, that coincides with $\R$ away from $\{ g=1\}$, is included in $\{ \Im z <0\}$ when $\Re z$ is inside $\{ g=1\}$,
and is
chosen in such a way that it stays below $\rho_{0}$ and at a distance $\sim h$ from it.

Then, setting $v=(v_1,v_2) := R_\theta (z) \varphi_\theta$, and denoting by $P_j^\theta$, $W_\theta$, $W^*_\theta$ the various distorted operators, we have,
$$\left\{
\begin{aligned}
& (P_1^\theta -z)v_1 + hW_\theta v_2 =\varphi_0^\theta ;\\
& (P_2^\theta -z)v_2 + hW^*_\theta v_1 = 0,
\end{aligned}
\right.
$$
and thus, for $z\in \gamma_-$,
\be
\label{v1v2}
\left\{
\begin{aligned}
& v_1 =\frac1{\lambda_0-z}\varphi_0^\theta -h(P_1^\theta -z)^{-1}W_\theta v_2  ;\\
& (1-M_\theta (z))v_2= \frac{-h}{\lambda_0-z}(P_2^\theta -z)^{-1}W^*_\theta \varphi_0^\theta,
\end{aligned}
\right.
\ee
with,
\be
\label{Mtheta}
M_\theta (z):= h^2(P_2^\theta -z)^{-1}W^*_\theta (P_1^\theta -z)^{-1}W_\theta.
\ee

In the next sections, we will prove that, for $h$ small enough, we have $\Vert M_\theta (z)\Vert <1$ (see \eqref{estMtheta}). Assuming for a while this result, we conclude from \eqref{v1v2} that we have,

\be
\label{v1v2bis}
\left\{
\begin{aligned}
& v_1 =\frac1{\lambda_0-z}\varphi_0^\theta +\frac{h^2}{\lambda_0-z}\sum_{\ell \geq 0}(P_1^\theta -z)^{-1}W_\theta M_\theta (z)^\ell (P_2^\theta -z)^{-1}W^*_\theta \varphi_0^\theta  ;\\
& (1-M_\theta (z))v_2= \sum_{\ell \geq 0}\frac{-h}{\lambda_0-z}M_\theta (z)^\ell (P_2^\theta -z)^{-1}W^*_\theta \varphi_0^\theta,
\end{aligned}
\right.
\ee

As a consequence, since $\la R_\theta (z)\phi_\theta, \phi_{-\theta}\ra=\la v_1, \varphi_0^{-\theta}\ra$, and $\la \varphi_0^\theta,\varphi_0^{-\theta}\ra = \Vert \varphi_0\Vert^2 =1$,
we obtain,
\be
\begin{aligned}
\la & R_\theta  (z)\phi_\theta, \phi_{-\theta}\ra \\
&  =\frac1{\lambda_0-z}+\frac{h^2}{(\lambda_0-z)^2}\sum_{\ell \geq 0}\la  M_\theta (z)^\ell (P_2^\theta -z)^{-1}W^*_\theta \varphi_0^\theta, W^*_{-\theta} \varphi_0^{-\theta}\ra.
\end{aligned}
\ee
Inserting into \eqref{reste}, we finally obtain,
\be
r(t,\phi, h) = \frac{h^2}{2i\pi}\sum_{\ell \geq 0}\int_{\gamma_-} \frac{e^{-itz} g(\Re z)}{(\lambda_0-z)^2}T_\ell (z) dz,
\ee
with
\be
\begin{aligned}
T_\ell (z):= \la  M_\theta (z)^\ell & (P_2^\theta -z)^{-1}W^*_\theta \varphi_0^\theta, W^*_{-\theta} \varphi_0^{-\theta}\ra\\
& -\la  M_{-\theta} (z)^\ell (P_2^{-\theta} -z)^{-1}W^*_{-\theta} \varphi_0^{-\theta}, W^*_{\theta} \varphi_0^{\theta}\ra.
\end{aligned}
\ee
Therefore, 
\be
r(t,\phi, h) =r_0(t,\phi, h) +r_1(t,\phi, h) +r_2(t,\phi, h) 
\ee
with,
\be
\label{r0r1}
\begin{aligned}
& r_0(t,\phi, h):= \frac{h^2}{2i\pi}\int_{\gamma_-} \frac{e^{-itz} g(\Re z)}{(\lambda_0-z)^2}T_0 (z) dz \\
& r_1(t,\phi, h):= \frac{h^2}{2i\pi}\int_{\gamma_-} \frac{e^{-itz} g(\Re z)}{(\lambda_0-z)^2}T_1 (z) dz \\
& r_2(t,\phi, h):=\frac{h^2}{2i\pi}\sum_{\ell \geq 2}\int_{\gamma_-} \frac{e^{-itz} g(\Re z)}{(\lambda_0-z)^2}T_\ell (z) dz,
\end{aligned}
\ee
and, by an additional change of contour of integration (that brings $\gamma_-$ onto $\R$), we also obtain,
$$
\begin{aligned}
& r_0(t,\phi, h)\\
& =\frac{h^2}{2i\pi}\lim_{\varepsilon \to 0_+}\int_{\R} \frac{e^{-itz} g (z)}{(\lambda_0+i\varepsilon-z)^2}\la  [(P_2-z-i0)^{-1}-(P_2-z+i0)^{-1}]W^*\varphi_0, W^*\varphi_0\ra dz,
\end{aligned}
$$
that is, by Stone's formula,
\be
r_0(t,\phi, h)=h^2\la e^{-itP_2}g(P_2)(P_2-\lambda_0-i0)^{-1}W^*\varphi_0, (P_2-\lambda_0+i0)^{-1}W^*\varphi_0\ra.
\ee

The next sections are devoted to the estimates on $\Vert M_{\pm \theta}(z)\Vert$ and on $r_0(t,\phi, h)$, $r_1(t,\phi, h)$ and $r_2(t,\phi, h)$.

\section{Fundamental solutions}
For $z\in {\mathcal D}_h(C_0)$ and $j=1,2$, let $u_{j,L}^\pm(z) =u_{j,L}^\pm (z,x) $ be the global WKB solutions to 
$(P_j - z)u = 0$ on $I_L:=(-\infty, 0]$ given, e.g., in \cite{FMW1} (in particular, $u_{j,L}^-(z)$ decays exponentially in $x$ at $-\infty$, while $u_{j,L}^+(z)$ grows exponentially). Let also $u_{j,R}^\pm(z)  =u_{j,R}^\pm (z,x) $ be the global WKB solutions to $(P_j - z)u = 0$ on a complex neighborhood of $I_R:=[0,+\infty)$, such that $u_{j,R}^-(z)$ decays exponentially in $x$ at infinity on $I_R^\theta:= \{ x+i\theta \nu(x)\,;\, x\geq 0\}$ ($\theta >0$ fixed small enough).

In particular, 
we see in  that $u_{2,R}^+(z)$ decays exponentially in $x$ at infinity on $I_R^{-\theta} = \overline{I_R^\theta}$.

We set,

\be
\label{eq1}
\begin{aligned}
K_{j,L}(z)[v](x) := & \frac{u_{j,L}^+(z,x)}{h^2 \W[u_{j,L}^+(z),u_{j,L}^-(z)]}  \int_{-\infty}^x  u_{j,L}^-(z,t)v(t)\,dt \\
& + \frac{u_{j,L}^-(z,x)}{h^2 \W[u_{j,L}^+(z),u_{j,L}^-(z)]} \int_x^{0}\!\!\!\! u_{j,L}^+(z,t)v(t)\,dt,
\end{aligned}
\ee
where  $v$ is in the space $C_b^0(I_L)$ of bounded continuous functions on $I_L$;

\be\label{eq1R+}
\begin{aligned}
K_{j,R}^+(z)[v](x) := &  \frac{u_{j,R}^-(z,x)}{h^2 \W[u_{j,R}^-(z),u_{j,R}^+(z)]} \int_{0}^x u_{j,R}^+(z,t)v(t)\,dt \\
&+ \frac{u_{j,R}^+(z,x)}{h^2 \W[u_{j,R}^-(z),u_{j,R}^+(z)]} \int_x^{+\infty}\!\!\!\! u_{j,R}^-(z,t)v(t)\,dt,
\end{aligned}
\ee
where  $v$ is in the space $C_b^0(I_R^+)$ of bounded continuous functions on $I_R^+:=I_R^\theta$, and the integrals run over $I_R^+$ (see \cite[Section 3.2]{FMW1});

\be\label{eq1R-}
\begin{aligned}
K_{j,R}^-(z)[v](x) := &  \frac{u_{j,R}^+(z,x)}{h^2 \W[u_{j,R}^+(z),u_{j,R}^-(z)]} \int_{0}^x u_{j,R}^-(z,t)v(t)\,dt \\
&+ \frac{u_{j,R}^-(z,x)}{h^2 \W[u_{j,R}^+(z),u_{j,R}^-(z)]} \int_x^{+\infty}\!\!\!\! u_{j,R}^+(z,t)v(t)\,dt,
\end{aligned}
\ee
where  $v$ is in the space $C_b^0(I_R^-)$ of bounded continuous functions on $I_R^-:=I_R^{-\theta}$, and the integrals run over $I_R^-$.

Then, as in \cite[Section 3]{FMW1}, we see that we have,
$$
\begin{aligned}
& (P_j - z)K_{j,L}(z) = 1\quad \mbox{on }\, C_b^0(I_L);\\
& (P_j - z)K_{j,R}^\pm(z) = 1\quad \mbox{on }\, C_b^0(I_R^\pm).
\end{aligned}
$$
In the sequel, we will need the following result:
\begin{proposition}\sl
\label{normK}
\be
\label{normconc}
\Vert K_{2,L}\Vert_{{\mathcal L}(L^2(I_L))} +\Vert K_{1,R}^\pm\Vert_{{\mathcal L}(L^2(I_R^\pm))}=\ord (h^{-\frac23});
\ee
\be
\label{normdeconc}
\Vert K_{1,L}\Vert_{{\mathcal L}(L^2(I_L))} +\Vert K_{2,R}^\pm\Vert_{{\mathcal L}(L^2(I_R^\pm))}=\ord (h^{-\frac76})
\ee
\end{proposition}
\begin{proof} The proofs on $I_L$ and on $I_R^\pm$ are very similar, so we just give the one on $I_L$. Since $\W[u_{j,L}^+(z),u_{j,L}^-(z)]\sim h^{-\frac23}$ ($j=1,2$), by \eqref{eq1} and the Schur Lemma (see, e,g,, \cite{Ma2}), it is enough to estimate,
$$
h^{-\frac43}\sup_{x\in I_L}\int_{j_L}|U_{j,L}(t,x)|dt
$$
with
$$
U_{j,L}(t,x):= u_{j,L}^+(x) u_{j,L}^-(t){\bf 1}_{t\leq x} + u_{j,L}^-(x) u_{j,L}^+(t){\bf 1}_{x\leq t}.
$$
When $x\leq x^*-\delta$ with $\delta >0$ fixed arbitrarily small, we know (see, e.g., \cite{FMW1}) that $U_{1,L}(t,x) =\ord (h^\frac13 e^{-c|x-t|/h})$ for some constant $c>0$. Hence, 
$$
h^{-\frac43}\sup_{x\leq x^*-\delta}\int_{I_L}|U_{1,L}(t,x)|dt=\ord (1).
$$
When $x^*-\delta \leq x\leq 0$, then $\int_{-\infty}^{x^*-2\delta}|U_{1,L}(t,x)|dt$ is exponentially small, while, for $t\in [x^*-2\delta, 0]$, we have,
$$
U_{1,L}(t,x) =\ord (h^\frac16 |t|^{-\frac14}|t-x^*|^{-\frac14}).
$$
We deduce,
$$
h^{-\frac43}\sup_{x^* -\delta\leq x\leq 0}\int_{I_L}|U_{1,L}(t,x)|dt=\ord (h^{-\frac76}),
$$
and thus,
\be
\Vert K_{1,L}\Vert_{{\mathcal L}(L^2(I_L))}=\ord (h^{-\frac76}).
\ee
Concerning $K_{2,L}$, the same estimate $U_{2,L}(t,x) =\ord (h^\frac13 e^{-c|x-t|/h})$ is valid for $x\leq -\delta$ ($\delta >0$ arbitrarily small, $c=c(\delta)>0$). Therefore,
$$
h^{-\frac43}\sup_{x\leq -\delta}\int_{I_L}|U_{2,L}(t,x)|dt=\ord (1).
$$
Now, if $x\in [-\delta, -Ch^{\frac23}]$  (with a constant $C>0$ sufficiently large), we can write,
$$
\begin{aligned}
h^{-\frac43}\int_{I_L}|U_{2,L}(t,x)|dt=&h^{-\frac43}\int_{-\infty}^{-2\delta}|U_{2,L}(t,x)|dt+h^{-\frac43}\int_{-2\delta}^{-Ch^\frac23}|U_{2,L}(t,x)|dt\\
&+h^{-\frac43}\int_{-Ch^\frac23}^0|U_{2,L}(t,x)|dt,
\end{aligned}
$$
where the first term of the right hand side is exponentially small, while the last term is $\ord(h^{-\frac23})$. The middle term can be estimated by,
$$
h^{-\frac43}\int_{-2\delta}^{-Ch^\frac23}|U_{2,L}(t,x)|dt=\ord(h^{-\frac43 +\frac16})\int_{-2\delta}^{-Ch^\frac23}
\frac{e^{-|t^{3/2} -x^{3/2}|/h}}{|t|^\frac14} dt,
$$
and first dividing the integral into $\int_{-2\delta}^{x} + \int_{x}^{-Ch^\frac23}$, then making the change of variable $t\mapsto -(ht)^\frac23$, we obtain,
$$
\begin{aligned}
h^{-\frac43}\int_{-2\delta}^{-Ch^\frac23}|U_{2,L}(t,x)|dt=& \ord(h^{-\frac43 +\frac16 +\frac23-\frac16})e^{{-|x|^\frac32/h}}\int_{C^\frac32}^{|x|^\frac32/h}
\frac{e^{t}}{\sqrt t} dt \\
& +\ord(h^{-\frac23})e^{{|x|^\frac32/h}} \int_{|x|^\frac32/h}^{(2\delta)^\frac32/h}\frac{e^{-t}}{\sqrt t} dt,
\end{aligned}
$$
and thus,
$$
h^{-\frac43}\int_{-2\delta}^{-Ch^\frac23}|U_{2,L}(t,x)|dt=\ord(h^{-\frac23 }).
$$
Finally, if $x\in [ -Ch^{\frac23}, 0]$, the same argument (but this time without dividing the integral $\int_{-2\delta}^{-Ch^\frac23}$) directly gives $h^{-\frac43}\int_{-2\delta}^{-Ch^\frac23}|U_{2,L}(t,x)|dt=\ord(h^{-\frac23 })$, and the estimate on $\Vert K_{2,L}\Vert_{{\mathcal L}(L^2(I_L))} $ follows. Similar arguments (but with $x^*$ substituted by some large enough value of $x$) also apply on $I_R^\pm$, and complete the proof of the proposition.
\end{proof}

\section{Resolvents}

We consider the space ${\mathcal S}$ of functions $ \varphi \in C^\infty (\R)$ that are analytic on $[x_\infty, +\infty)$ and admit a holomorphic extension (still denoted by $\varphi$) near $\Gamma_\delta :=\{ x\in \C\, ;\, \Re x \geq x_\infty,\, |\Im x| \leq \delta \Re x\}$ for some $\delta >0$, and that are exponentially small at infinity both on $\R_-$ and on $\Gamma_\delta$.

In particular, for all $\varphi \in {\mathcal S}$, we have $K_{1,R}^+(z)[\varphi]=K_{1,R}^-(z)[\varphi]=:K_{1,R}(z)[\varphi]$ on $\R_+\cup \Gamma_\delta$.

For $z\in {\mathcal D}_h(C_0)\cap \{ \pm\Im z > 0\}$ and $j=1,2$, we denote  by $R_j^\pm(z) = (P_j-z)^{-1}$ the resolvent of $P_j$ in $z$, referred to as the {\sl incoming} (respectively {\sl out-going}) resolvent of $P_j$ in $z$.

Then, for $\varphi \in {\mathcal S}$, the next  Proposition  will show that $R_j^\pm (z)\varphi$ extend analytically to $z\in {\mathcal D}_h(C_0)$ ($z\notin {\rm Sp}(P_1)$ in the case $j=1$), and we use the same notations for their extensions. Obviously, in the case $j=1$, one also has $R_1^+(z)\varphi =R_1^-(z)\varphi$ for $z\in {\mathcal D}_h(C_0)\backslash {\rm Sp}(P_1)$.

Finally, for $\varphi \in {\mathcal S}$, we denote by $\varphi_L$ its restriction to $I_L$ and by $\varphi_R$ its restriction to $\R_+\cup \Gamma_\delta$.

\begin{proposition}\sl
\label{propResolv}
(i) For all $\varphi \in {\mathcal S}$, $z\in {\mathcal D}_h(C_0)\backslash {\rm Sp}(P_1)$, and $x\leq 0$, one has,
$$
R_1(z)\varphi (x)= K_{1,L}(z)[\varphi_L ](x)+ \alpha_{L}(z)[\varphi]u_{1,L}^-(z,x),
$$
with,
$$
\alpha_L(z)[\varphi] = \alpha_{L,L}(z)[\varphi_L]+ \alpha_{L,R}(z)[\varphi_R],
$$

$$
\begin{aligned}
 & \alpha_{L,L}(z)[\varphi_L]:= \frac{h^{-2}\W(u_{1,L}^+(z), u_{1,R}^-(z))}{ \W(u_{1,R}^-(z), u_{1,L}^-(z))\W(u_{1,L}^+(z), u_{1,L}^-(z))}\int_{-\infty}^0u_{1,L}^-(z,t)\varphi_L(t)dt\\
  & \alpha_{L,R}(z)[\varphi_R]:=\frac{h^{-2}}{ \W(u_{1,R}^-(z), u_{1,L}^-(z))}\int_0^{+\infty}u_{1,R}^-(z,t)\varphi_R(t)dt.
\end{aligned}
$$

(ii) For all $\varphi \in {\mathcal S}$, $z\in {\mathcal D}_h(C_0)\backslash {\rm Sp}(P_1)$, and $x\in \R_+\cup \Gamma_\delta$, one has,
$$
R_1(z)\varphi (x)= K_{1,R}(z)[\varphi_R ](x)+ \alpha_{R}(z)[\varphi]u_{1,R}^-(z,x),
$$
with,
$$
\alpha_R(z)[\varphi] = \alpha_{R,L}(z)[\varphi_L]+ \alpha_{R,R}(z)[\varphi_R],
$$

$$
\begin{aligned}
 & \alpha_{R,L}(z)[\varphi_L]:=  \frac{h^{-2}}{ \W(u_{1,R}^-(z), u_{1,L}^-(z))}\int_{-\infty}^0u_{1,L}^-(z,t)\varphi_L(t)dt\\
  & \alpha_{R,R}(z)[\varphi_R]:= \frac{h^{-2}\W(u_{1,L}^-(z), u_{1,R}^+(z))}{ \W(u_{1,R}^-(z), u_{1,L}^-(z))\W(u_{1,R}^-(z), u_{1,
  R}^+(z))}\int_0^{+\infty}u_{1,R}^-(z,t)\varphi_R(t)dt.
\end{aligned}
$$

(iii) For all $\varphi \in {\mathcal S}$, $z\in {\mathcal D}_h(C_0)$, and $x\leq 0$, one has,
$$
R_2^\pm (z)\varphi (x)= K_{2,L}(z)[\varphi_L ](x)+ \beta_{L}^\pm (z)[\varphi]u_{2,L}^-(z,x),
$$
with,
$$
\beta_{L}^\pm (z)[\varphi] = \beta_{L,L}^\pm(z)[\varphi_L]+ \beta_{L,R}^\pm(z)[\varphi_R],
$$

$$
\begin{aligned}
 & \beta_{L,L}^\pm(z)[\varphi_L]:= \frac{h^{-2}\W(u_{2,L}^+(z), u_{2,R}^\mp (z))}{ \W(u_{2,R}^\mp(z), u_{2,L}^-(z))\W(u_{2,L}^+(z), u_{2,L}^-(z))}\int_{-\infty}^0u_{2,L}^-(z,t)\varphi_L(t)dt\\
  & \beta_{L,R}^\pm(z)[\varphi_R]:=\frac{h^{-2}}{ \W(u_{2,R}^\mp(z), u_{2,L}^-(z))}\int_{I_R^\pm}u_{2,R}^\mp(z,t)\varphi_R(t)dt.
\end{aligned}
$$

(iv) For all $\varphi \in {\mathcal S}$, $z\in {\mathcal D}_h(C_0)$, and $x\in I_R^\pm$, one has,
$$
R_2^\pm(z)\varphi (x)= K_{2,R}^\pm (z)[\varphi_R ](x)+ \beta_{R}^\pm(z)[\varphi]u_{2,R}^\mp(z,x),
$$
with,
$$
\beta_R^\pm(z)[\varphi] = \beta_{R,L}^\pm(z)[\varphi_L]+ \beta_{R,R}^\pm(z)[\varphi_R],
$$
$$
\begin{aligned}
 & \beta_{R,L}^\pm(z)[\varphi_L]:=  \frac{h^{-2}}{ \W(u_{2,R}^\mp(z), u_{2,L}^-(z))}\int_{-\infty}^0u_{2,L}^-(z,t)\varphi_L(t)dt\\
  & \beta_{R,R}^\pm(z)[\varphi_R]:= \frac{h^{-2}\W(u_{2,L}^-(z), u_{2,R}^\pm(z))}{ \W(u_{2,R}^\mp(z), u_{2,L}^-(z))\W(u_{2,R}^\mp (z), u_{2,
  R}^\pm(z))}\int_{I_R^\pm}u_{2,R}^\mp(z,t)\varphi_R(t)dt.
\end{aligned}
$$
\end{proposition}

\begin{remark}\sl
\label{rem52}
In particular, by \cite[Appendix]{FMW1}, for $z\in \gamma_-$ we have,
$$
\begin{aligned}
& \alpha_{L,L}(z)[\varphi_L] = \left(\frac{\pi}4 h^{-\frac43}\tan \frac{{\mathcal A}(z)}{h} +\ord(h^{-1})\right)\int_{-\infty}^0u_{1,L}^-(z,t)\varphi_L(t)dt ;\\
& \alpha_{L,R}(z)[\varphi_R] = \left(\frac{\pi}4 h^{-\frac43}\left (\cos \frac{{\mathcal A}(z)}{h}\right)^{-1}+\ord(h^{-1})\right)\int_0^{+\infty}u_{1,R}^-(z,t)\varphi_R(t)dt ;\\
& \alpha_{R,L}(z)[\varphi_L] = \left(\frac{\pi}4 h^{-\frac43}\left (\cos \frac{{\mathcal A}(z)}{h}\right)^{-1}+\ord(h^{-1})\right)\int_{-\infty}^0u_{1,L}^-(z,t)\varphi_L(t)dt ;\\
& \alpha_{R,R}(z)[\varphi_L] = \left(\frac{\pi}4 h^{-\frac43}\tan \frac{{\mathcal A}(z)}{h}+\ord(h^{-1})\right)\int_0^{+\infty}u_{1,R}^-(z,t)\varphi_R(t)dt ;\\
& \beta_{L,L}^\pm(z)[\varphi_L] = \left(\pm i \frac{\pi}4 h^{-\frac43}+\ord(h^{-1})\right)\int_{-\infty}^0u_{2,L}^-(z,t)\varphi_L(t)dt ;\\
& \beta_{L,R}^+(z)[\varphi_R] = \left(\frac{\pi}{\sqrt 2}e^{i\frac{\pi}4}h^{-\frac43}+\ord(h^{-1})\right)\int_{I_R^+}u_{2,R}^-(z,t)\varphi_R(t)dt ;\\
& \beta_{L,R}^-(z)[\varphi_R] = -\left(\frac{\pi}{2\sqrt 2}e^{i\frac{\pi}4}h^{-\frac43}+\ord(h^{-1})\right)\int_{I_R^-}u_{2,R}^+(z,t)\varphi_R(t)dt ;\\
& \beta_{R,L}^+(z)[\varphi_L] = \left(\frac{\pi}{\sqrt 2}e^{i\frac{\pi}4} h^{-\frac43}+\ord(h^{-1})\right)\int_{-\infty}^0u_{2,L}^-(z,t)\varphi_L(t)dt ;\\
& \beta_{R,L}^-(z)[\varphi_L] = -\left(\frac{\pi}{2\sqrt 2}e^{i\frac{\pi}4}h^{-\frac43}+\ord(h^{-1})\right)\int_{-\infty}^0u_{2,L}^-(z,t)\varphi_L(t)dt ;\\
& \beta_{R,R}^+(z)[\varphi_R] = \left(\pi h^{-\frac43}+\ord(h^{-1})\right)\int_{I_R^+}u_{2,R}^-(z,t)\varphi_R(t)dt ;\\
& \beta_{R,R}^-(z)[\varphi_R] = -\left(\frac{\pi}4 h^{-\frac43}+\ord(h^{-1})\right)\int_{I_R^-}u_{2,R}^+(z,t)\varphi_R(t)dt,
\end{aligned}
$$
with,
$$
{\mathcal A}(z):= \int_{x_1^*(z)}^{x_1(z)}\sqrt{ z-V_1(t)} \, dt,
$$
where $x_1^*(z)$ (respectively $x_1(z)$) is the unique solution of $V_1(x)=z$ close to $x^*$ (respectively close to 0).
\end{remark}

\begin{proof} We only prove (i)-(ii), since (iii)-(iv) follow along the same lines. We set,
$$
\psi:= R_1(z)\varphi\quad ;\quad \psi_{1,L}:= K_{1,L}(z)\varphi \quad ;\quad \psi_{1,R}:= K_{1,R}(z)\varphi.
$$
Then by construction we have,
$$
\begin{aligned}
& (P_1-z)(\psi -\psi_{1,L}) = (P_1-z)(\psi -\psi_{1,R}) =0;\\
 & \psi -\psi_{1,L}\in L^2(I_L)\quad ;\quad \psi -\psi_{1,R}\in L^2(I_R^\pm).
 \end{aligned}
$$
Therefore, there exist two complex numbers $\alpha_L =\alpha_L(z)$ and $\alpha_R =\alpha_R(z)$ such that,
$$
 \psi -\psi_{1,L}=\alpha_L u_{1,L}^-\quad ;\quad \psi -\psi_{1,R}=\alpha_R u_{1,R}^-.
$$
In order to compute $\alpha_L$ and $\alpha_R$, we write that $\psi$ must be $C^1$ at 0. We find the system,
$$
\left\{
\begin{aligned}
&\alpha_L u_{1,L}^-(0)- \alpha_R u_{1,R}^-(0)= \psi_{1,R}(0)- \psi_{1,L}(0);\\
&\alpha_L [u_{1,L}^-]'(0)- \alpha_R [u_{1,R}^-]'(0)=\psi_{1,R}'(0)-\psi_{1,L}'(0),
\end{aligned}
\right.
$$
and, using that, 
$$
\begin{aligned}
&\psi_{1,L}(0)=\frac{u_{1.L}^+(0)}{h^2\W(u_{1,L}^+, u_{1,L}^-)}\int_{-\infty}^0 u_{1,L}^-(t)\varphi(t)dt;\\
&\psi_{1,L}'(0)=\frac{[u_{1.L}^+]'(0)}{h^2\W(u_{1,L}^+, u_{1,L}^-)}\int_{-\infty}^0 u_{1,L}^-(t)\varphi(t)dt;\\
&\psi_{1,R}(0)=\frac{u_{1.R}^+(0)}{h^2\W(u_{1,R}^-, u_{1,R}^+)}\int_0^{+\infty} u_{1,R}^-(t)\varphi(t)dt;\\
&\psi_{1,R}'(0)=\frac{[u_{1.R}^+]'(0)}{h^2\W(u_{1,R}^-, u_{1,R}^+)}\int_0^{+\infty} u_{1,R}^-(t)\varphi(t)dt,
\end{aligned}
$$
the result follows by straightforward computations.
\end{proof}
As a consequence of the previous proposition, we have,
\begin{corollary}\sl 
\label{prop5.3}
For $z\in \gamma_-$, one has,
\be
\label{estR2}
 \Vert  R_2^\pm (z)\Vert_{{\mathcal L}(L^2(I_L\cup I_R^\pm ))} = {\mathcal O}(h^{-1-1/6});
\ee
\be
\label{estM}
\Vert h^2 R_2^\pm (z)W^* R_1(z)W\Vert_{{\mathcal L}(L^2(I_L\cup I_R^\pm ))} = {\mathcal O}(h^{1/6}).
\ee
\end{corollary}
\begin{remark}\sl
In particular, the operators $M_{\pm\theta}$ introduced in \eqref{Mtheta} also satisfy,
\be
\label{estMtheta}
\Vert M_{\pm\theta}(z) \Vert_{{\mathcal L}(L^2(\R ))}= {\mathcal O}(h^{1/6}).
\ee
\end{remark}
\begin{proof} We first observe that, by construction (see, e.g., \cite{FMW1}), we have,
$$
\Vert u_{2,L}^-\Vert_{L^2(I_L)}=\ord(h^{\frac13})\quad ;\quad  \Vert u_{2,R}^\mp\Vert_{L^2(I_R^\pm)}=\ord(h^{\frac16}).
$$
Using Remark \ref{rem52}, we deduce that, for $z\in\gamma_-$ and $S=L,R$, we have,
$$
\begin{aligned}
& |\beta_{S,L}^\pm(z)[\varphi_L] |=\ord(h^{-1})\Vert \varphi_L\Vert_{L^2(I_L)} \\
& |\beta_{S,R}^\pm(z)[\varphi_R] |=\ord(h^{-\frac76})\Vert \varphi_R\Vert_{L^2(I_R^\pm)}.
\end{aligned}
$$
Therefore,
\be
\label{betau}
\Vert \beta_L^\pm(z)[\varphi] u_{2,L}^-\Vert_{L^2(I_L)}+\Vert \beta_R^\pm(z)[\varphi] u_{2,R}^\mp\Vert_{L^2(I_R^\pm)}=\ord(h^{-1})\Vert \varphi\Vert_{L^2(I_L\cup I_R^\pm)}.
\ee
Moreover, by Proposition \ref{normK}, we also have,
$$
\Vert K_{2,L}(z)[\varphi_L]\Vert_{L^2(I_L)} +\Vert K_{2,R}^\pm(z)[\varphi_R]\Vert_{L^2(I_R^\pm)}=\ord (h^{-\frac76})\Vert\varphi_R\Vert_{L^2(I_R^\pm)}.
$$
Then,  \eqref{estR2} follows from Proposition \ref{propResolv} (iii)-(iv).

Concerning \eqref{estM}, in order to simplify the notations we write the detailed proof with $W=W^*={\mathbf 1}$, and then we explain how to deduce the result for the actual $W$, $W^*$. We set,
$$
K_0^\pm(z):=h^2 R_2^\pm (z)R_1(z),
$$
and we first observe that, for $z\in\gamma_-$, we have $\Vert R_1(z)\Vert =\ord (h^{-1})$, so that, by \eqref{estR2}, a mere estimate with the product of the norms gives $\Vert K_0^\pm(z)\Vert =\ord (h^{-\frac16})$. The improvement into $\ord(h^{\frac16})$ will actually follow from the fact that $R_2(z)$ is better on $I_L$, while $R_1(z)$ is better on $I_R^\pm$.

Using Proposition \ref{propResolv} (and dropping the parameter $z$), we have,
\begin{itemize}
\item On $I_L$,
\be
\label{K0IL}
\begin{aligned}
K_0^\pm\varphi =& h^2K_{2,L}K_{1,L}\varphi_L+ h^2\alpha_{L}[\varphi]K_{2,L}u_{1,L}^-+ h^2\beta_{L,L}^\pm [K_{1,L}\varphi_L]u_{2,L}^-\\
& + h^2 \alpha_{L}[\varphi]\beta_{L,L}^\pm [u_{1,L}^-]u_{2,L}^-+ h^2\beta_{L,R}^\pm [K_{1,R}\varphi_R]u_{2,L}^-\\
&+ h^2\alpha_{R}[\varphi]\beta_{L,R}^\pm [u_{1,R}^-]u_{2,L}^- ;
\end{aligned}
\ee
\item On $I_R^\pm$,
\be
\begin{aligned}
K_0^\pm\varphi =& h^2K_{2,R}K_{1,R}\varphi_R+ h^2\alpha_{R}[\varphi]K_{2,R}u_{1,R}^-+ h^2\beta_{R,L}^\pm [K_{1,L}\varphi_L]u_{2,R}^\mp\\
& + h^2 \alpha_{L}[\varphi]\beta_{R,L}^\pm [u_{1,L}^-]u_{2,R}^\mp + h^2\beta_{R,R}^\pm [K_{1,R}\varphi_R]u_{2,R}^\mp\\
&+ h^2\alpha_{R}[\varphi]\beta_{R,R}^\pm [u_{1,R}^-]u_{2,R}^\mp.
\end{aligned}
\ee
\end{itemize}
Since the studies on $I_L$ and on $I_R^\pm$ are similar, we detail the proof for $I_L$ only. In view of \eqref{K0IL}, we have six terms to examine. We first show,
\begin{lemma}\sl 
\label{lemm55}
One has,
\be
\Vert h^2 \alpha_{L}[\varphi]\beta_{L,L}^\pm [u_{1,L}^-]u_{2,L}^-\Vert_{L^2(I_L)}=\ord (h^\frac12)\Vert \varphi\Vert_{L^2(I_L\cup I_R^\pm)};
\ee
\be
\Vert h^2\alpha_{R}[\varphi]\beta_{L,R}^\pm [u_{1,R}^-]u_{2,L}^- \Vert_{L^2(I_L)}=\ord (h^\frac12)\Vert \varphi\Vert_{L^2(I_L\cup I_R^\pm)}.
\ee
\end{lemma}
\begin{proof} 
Since $\Vert u_{1,L}^-\Vert_{L^2(I_L)}$ and $\Vert u_{2,R}^\mp\Vert_{L^2(I_R^\pm)}$ are  of size $\sim  h^{\frac16}$,  while  $\Vert u_{1,R}^-\Vert_{L^2(I_R^\pm)}$ and $\Vert u_{2,L}^-\Vert_{L^2(I_L)}$ are  of size $\sim  h^{\frac13}$, by Cauchy-Schwarz inequality we see on Remark \ref{rem52} that we have,
$$
\Vert h^2 \alpha_{L}[\varphi]u_{2,L}^-\Vert_{L^2(I_L)}+ \Vert h^2\alpha_{R}[\varphi] u_{2,L}^- \Vert_{L^2(I_L)}=\ord (h^\frac76)\Vert \varphi\Vert_{L^2(I_L\cup I_R^\pm)},
$$
and it remains to estimate $\left|\beta_{L,L}^\pm [u_{1,L}^-]\right|$ and $\left|\beta_{L,R}^\pm [u_{1,R}^-]\right|$. Both can be upper bounded by,
$$
Ch^{-\frac43}\int_0^{Ch^\frac23}dt + Ch^{-\frac43}\int_{Ch^\frac23}^\delta \frac{h^\frac13}{\sqrt t}e^{-ct^{3/2}/h}dt+Ch^{-\frac43}\int_\delta^{+\infty}h^\frac13 e^{-ct/h},
$$
(with $C>0$ large enough constant, and $c,\delta >0$ small enough constants), and thus are $\ord(h^{-\frac23})$, and the result follows.
\end{proof}
\begin{lemma}\sl
$$
\Vert h^2\alpha_{L}[\varphi]K_{2,L}u_{1,L}^-\Vert_{L^2(I_L)}=\ord (h^\frac13)\Vert \varphi\Vert_{L^2(I_L\cup I_R^\pm)}.
$$
\end{lemma}
\begin{proof} Using again Remark \ref{rem52}, we have,
$$
\Vert h^2\alpha_{L}[\varphi]K_{2,L}u_{1,L}^-\Vert_{L^2(I_L)}=\ord (h^\frac56)\Vert \varphi\Vert_{L^2(I_L\cup I_R^\pm)}\Vert K_{2,L}u_{1,L}^-\Vert_{L^2(I_L)},
$$
and it remains to estimate $\Vert K_{2,L}u_{1,L}^-\Vert_{L^2(I_L)}$. Applying Proposition \ref{normK}, we obtain,
$$
\Vert K_{2,L}u_{1,L}^-\Vert_{L^2(I_L)}=\ord (h^{-\frac23})\Vert u_{1,L}^-\Vert_{L^2(I_L)}=\ord (h^{-\frac12}),
$$
and the result follows.
\end{proof}

\begin{lemma}\sl
$$
\begin{aligned}
 \Vert h^2\beta_{L,L}^\pm [K_{1,L}\varphi_L]u_{2,L}^- \Vert_{L^2(I_L)}+ \Vert h^2\beta_{L,R}^\pm  [K_{1,R}\varphi_R] & u_{2,L}^-\Vert_{L^2(I_L)}\\
 & =\ord (h^\frac16)\Vert \varphi \Vert_{L^2(I_L\cup I_R^\pm)}.
 \end{aligned}
 $$
\end{lemma}
\begin{proof} Since $\Vert u_{2,L}^-\Vert_{L^2(I_L)}=\ord (h^\frac13)$, it is enough to prove that $\beta_{L,L}^\pm [K_{1,L}\varphi_L]$ and $\beta_{L,R}^\pm  [K_{1,R}\varphi_R]$ are $\ord (h^{-\frac{13}6})\Vert \varphi \Vert_{L^2(I_L\cup I_R^\pm)}$. By Remark \ref{rem52} and Proposition \ref{normK}, we have,
$$
\beta_{L,L}^\pm [K_{1,L}\varphi_L]=\ord (h^{-1})\Vert K_{1,L}\varphi_L\Vert_{L^2(I_L)}=\ord (h^{-1-\frac76})\Vert \varphi_L\Vert_{L^2(I_L)},
$$
$$
\beta_{L,R}^\pm  [K_{1,R}\varphi_R]=\ord (h^{-\frac76})\Vert K_{1,R}\varphi_R\Vert_{L^2(I_R^\pm)}=\ord (h^{-\frac76-\frac23})\Vert \varphi_R\Vert_{L^2(I_R^\pm)},
$$
and the result follows.
\end{proof}

\begin{lemma}\sl
\label{lemm58}
$$
\Vert h^2K_{2,L}K_{1,L}\Vert_{{\mathcal L}(L^2(I_L))}=\ord (h^\frac16).
$$
\end{lemma}
\begin{proof} It is an immediate consequence of Proposition \ref{normK}.
\end{proof}

Using \eqref{K0IL}, we conclude from Lemmas \ref{lemm55}-\ref{lemm58} that we have,
$$
\Vert K_0^\pm\varphi\Vert_{L^2(I_L)} =\ord (h^\frac16)\Vert \varphi\Vert_{L^2(I_L\cup I_R^\pm)}.
$$
Analogous arguments lead to the same estimate on $I_R^\pm$, and thus, we have proved,
$$
\Vert K_0^\pm\Vert_{{\mathcal L}(L^2(I_L\cup I_R^\pm))}=\ord(h^\frac16).
$$
Concerning the result for $\Vert h^2R_2^\pm (z) W^*R_1(z)W\Vert$, we first observe that the previous proof works without changes for $\Vert h^2R_2^\pm (z) f R_1(z)g\Vert$ if $f,g$ are bounded multiplication operators, and also for $\Vert h^2R_2^\pm (z) f hD_x R_1(z)g\Vert$ because the estimates on $h[u_{1,L}^\pm]'$ and $h[u_{1,R}^\pm]'$ are the same as (and at some places even better than) those on $u_{1,L}^\pm$ and $u_{1,R}^\pm$.

Then, \eqref{estM} can easily be deduced by writing,
$$
R_1(z)hD_x=\left( I+R_1(z)(I+z-V_1)\right) hD_x (1-h^2\Delta)^{-1}
$$
(and the analogous formula for $R_2^\pm (z)$), and by using \eqref{estR2} and the fact that $\Vert hD_x (1-h^2\Delta)^{-1} \Vert =\ord (1)$.
\end{proof}

\section{Function spaces}

In order to estimate in a systematic way the various integrals that are involved in the expressions of $r_0(t,\phi, h)$, $r_1(t,\phi, h)$ and $r_2(t,\phi, h)$, we introduce several function spaces that, in some way, are related to the behavior (both semiclassical and at infinity in $x$) of the global WKB solutions of the scalar problems.

We set,
$$
\begin{aligned}
& m_0(x)=m_0(x;h):= \min (h^{-1/6}, |x|^{-1/4});\\
& m_*(x)=m_*(x;h):= \min (h^{-1/6}, |x-x^*|^{-1/4}).
\end{aligned}
$$
            
We define the space ${\mathcal F}_1(I_L)$ as the space of $h$-dependent smooth functions $u=u(x;h)$ on $I_L$ for which, for any $\delta >0$ small enough and for any $k\geq 0$, there exists a constant $c=c_{k ,\delta} >0$ such that,
\begin{itemize}
\item On $(-\infty, x^*-\delta]$, $(hD_x)^k u(x;h) = {\mathcal O}(e^{-c |x|/h})$;
\item On $[x^*-\delta, x^*]$, $(hD_x)^k u(x;h) = {\mathcal O}(m_*(x)e^{-c |x-x^*|^{3/2}/h})$;
\item On $[x^*, x^*+\delta]$, $(hD_x)^k u(x;h) = {\mathcal O}(m_*(x))$;
\item On $[ x^*+\delta, -\delta]$, $(hD_x)^k u(x;h) = {\mathcal O}(1)$;
\item On $[ -\delta, 0]$, $(hD_x)^k u(x;h) = {\mathcal O}( m_0(x))$.
\end{itemize}

We also define the space ${\mathcal F}_2(I_L)$ as the space of $h$-dependent smooth functions $u=u(x;h)$ on $I_L$ for which, for any $\delta >0$ small enough and for any $k\geq 0$, there exists a constant $c=c_{k ,\delta} >0$ such that,
\begin{itemize}
\item On $(-\infty, -\delta]$, $(hD_x)^k u(x;h) = {\mathcal O}(e^{-c |x|/h})$;
\item On $[-\delta, 0]$, $(hD_x)^k u(x;h) = {\mathcal O}(m_0(x)e^{-c |x|^{3/2}/h})$.
\end{itemize}

Analogously, we define the space ${\mathcal F}_1(I_R^\pm)$ as the space of $h$-dependent smooth functions $u=u(x;h)$ on $I_R^\pm$ for which, for any $\delta >0$ and for any $k\geq 0$, there exists a constant $c=c_{k ,\delta} >0$ such that,
\begin{itemize}
\item On $[0, \delta]$, $(hD_x)^k u(x;h) = {\mathcal O}(m_0(x)e^{-c |x|^{3/2}/h})$;
\item On $[\delta, +\infty)$, $(hD_x)^k u(x;h) = {\mathcal O}(e^{-c |x|/h})$.
\end{itemize}

Finally, we define the space ${\mathcal F}_2(I_R^\pm)$ as the space of $h$-dependent smooth functions $u=u(x;h)$ on $I_R^\pm$ for which, for any $\delta >0$ and for any $k\geq 0$, there exist two constants $c=c_{k ,\delta} >0$ and $C= C_{k, \delta }>0$ such that,
\begin{itemize}
\item On $[ 0, \delta ]$, $(hD_x)^k u(x;h) = {\mathcal O}( m_0(x))$;
\item On $I_R^\pm \cap \{ \delta \leq \Re x \leq C\} $, $(hD_x)^k u(x;h) = {\mathcal O}(1)$;
\item On $I_R^\pm \cap \{  \Re x \geq C\} $, $(hD_x)^k u(x;h) = {\mathcal O}(e^{-c |x|/h})$.
\end{itemize}

For $j,k\in \{1,2\}$, we also denote by ${\mathcal F}_j(I_L)\cap {\mathcal F}_k(I_R^\pm)$ the space of $h$-dependent functions $\varphi$ defined on $I_L\cup I_R^\pm$ (not necessarily smooth at 0), such that $\varphi_L\in {\mathcal F}_j(I_L)$ and $\varphi_R\in {\mathcal F}_k(I_R^\pm)$. Of course, if such a function $\varphi$ is smooth at 0, then, for any $\ell\geq 0$, one also has $(hD_x)^\ell\varphi \in {\mathcal F}_j(I_L)\cap {\mathcal F}_k(I_R^\pm)$.

In particular, for any $z\in {\mathcal D}_h(C_0)$, we can see,
\be
\label{ujLRin}
\begin{aligned}
& u_{1,L}^-(z) \in h^{1/6}{\mathcal F}_1(I_L)\quad ; \quad u_{1,R}^-(z) \in h^{1/6}{\mathcal F}_1(I_R^\pm);\\
& u_{2,L}^-(z) \in h^{1/6}{\mathcal F}_2(I_L)\quad ; \quad u_{2,R}^\mp(z) \in h^{1/6}{\mathcal F}_2(I_R^\pm),
\end{aligned}
\ee
and also, since $\varphi_0 \sim h^{-1/6}u_{1,L}^-(\lambda_0)$ on $\R_-$, and $\varphi_0  \sim h^{-1/6}u_{1,R}^-(\lambda_0)$ on $I_R^\pm$,
\be
\label{phi2in}
\varphi_0 \in {\mathcal F}_1(I_L)\cap {\mathcal F}_1(I_R^\pm)\cap C^\infty,
\ee
where $C^\infty$ stands for $C^\infty (I_L\cup I_R^\pm)$, and just means that $\varphi_0$ is smooth at 0, too.

We have (dropping the $z$-dependence),
\begin{proposition}\sl 
\label{actK}
The following inclusions hold:
$$
\begin{aligned}
& K_{1,L}\left( {\mathcal F}_1(I_L)\right) \subset  h^{-1}{\mathcal F}_1(I_L)\quad  ; \quad K_{1,L}\left( {\mathcal F}_2(I_L)\right) \subset  h^{-2/3}{\mathcal F}_1(I_L);\\
& K_{2,L}\left( {\mathcal F}_1(I_L)\right) \subset h^{-2/3}{\mathcal F}_1(I_L) \quad ; \quad
 K_{2,L}\left( {\mathcal F}_2(I_L)\right) \subset h^{-2/3}{\mathcal F}_2(I_L ),
\end{aligned}
$$

and,
$$
\begin{aligned}
& K_{1,R}\left( {\mathcal F}_1(I_R^\pm)\right) \subset  h^{-2/3}{\mathcal F}_1(I_R^\pm)\quad  ; \quad K_{1,R}\left( {\mathcal F}_2(I_R^\pm)\right) \subset  h^{-2/3}{\mathcal F}_2(I_R^\pm);\\
& K_{2,R}^\pm\left( {\mathcal F}_1(I_R^\pm)\right) \subset h^{-2/3}{\mathcal F}_2(I_R^\pm ) \quad  ;\quad K_{2,R}^\pm\left( {\mathcal F}_2(I_R^\pm)\right) \subset  h^{-1}{\mathcal F}_2(I_R^\pm ).
\end{aligned}
$$

\end{proposition}
\begin{proof} See Appendix 1.
\end{proof}

\begin{remark}\sl As an immediate consequence of the definitions of ${\mathcal F}_j(I_L)$ and ${\mathcal F}_j(I_R^\pm)$ ($j=1,2$), we have
\begin{itemize}
\item If $v\in {\mathcal F}_1(I_L)$, then $\Vert v\Vert_{L^2(I_L)} = \ord (1)$;
\item If $v\in {\mathcal F}_2(I_L)$, then $\Vert v\Vert_{L^2(I_L)} = \ord (h^\frac16)$;
\item If $v\in {\mathcal F}_1(I_R^\pm)$, then $\Vert v\Vert_{L^2(I_R^\pm)} = \ord (h^\frac16)$;
\item If $v\in {\mathcal F}_2(I_R^\pm)$, then $\Vert v\Vert_{L^2(I_R^\pm)} = \ord (1)$,
\end{itemize}
uniformly as $h\to 0_+$.
\end{remark}

\begin{proposition}\sl
\label{actalphabeta} 

For $z\in \gamma_-$, one has,
$$
\begin{aligned}
& |\alpha_{L,L}(z)| + |\alpha_{R,L}(z)| = {\mathcal O}( h^{-7/6}) \, \mbox{ on } {\mathcal F}_1(I_L);\\
& |\beta^\pm_{L,L}(z)| + |\beta^\pm_{R,L}(z)| = {\mathcal O}( h^{-5/6}) \, \mbox{ on } {\mathcal F}_1(I_L);
\end{aligned}
$$

$$
\begin{aligned}
& |\alpha_{L,L}(z)| + |\alpha_{R,L}(z)| = {\mathcal O}( h^{-5/6}) \, \mbox{ on } {\mathcal F}_2(I_L);\\
& |\beta^\pm_{L,L}(z)| + |\beta^\pm_{R,L}(z)| = {\mathcal O}(h^{-5/6}) \, \mbox{ on } {\mathcal F}_2(I_L);
\end{aligned}
$$

$$
\begin{aligned}
& |\alpha_{L,R}(z)| + |\alpha_{R,R}(z)| = {\mathcal O}( h^{-5/6}) \, \mbox{ on } {\mathcal F}_1(I_R^\pm);\\
& |\beta^\pm_{L,R}(z)| + |\beta^\pm_{R,R}(z)| = {\mathcal O}(h^{-5/6}) \, \mbox{ on } {\mathcal F}_1(I_R^\pm);
\end{aligned}
$$

$$
\begin{aligned}
& |\alpha_{L,R}(z)| + |\alpha_{R,R}(z)| = {\mathcal O}( h^{-5/6}) \, \mbox{ on } {\mathcal F}_2(I_R^\pm);\\
& |\beta^\pm_{L,R}(z)| + |\beta^\pm_{R,R}(z)| = {\mathcal O}( h^{-7/6}) \, \mbox{ on } {\mathcal F}_2(I_R^\pm).
\end{aligned}
$$

\end{proposition}
\begin{proof} We use Remark \ref{rem52}.
Since $u_{1,L}^-\in h^{\frac16}{\mathcal F}_1(I_L)$, for $\varphi\in {\mathcal F}_1(I_L)$ we have,
$$
|\alpha_{L,L}(z)[\varphi]| +|\alpha_{R,R}(z)[v]| = \ord (h^{-\frac43})\Vert u_{1,L}^-\Vert_{L^2(I_L)}\Vert \varphi\Vert_{L^2(I_L)}=\ord(h^{-\frac43 +\frac16}).
$$
Since $u_{2,L}^-$ is exponentially concentrated at $x=0$, we also have,
$$
|\beta^\pm_{L,R}(z)[\varphi]| + |\beta^\pm_{R,R}(z)[\varphi]|= \ord(h^{-\frac43})\int_0^\delta \frac{h^{\frac16}e^{-ct^{3/2}/h}}{\sqrt t}dt + \ord(e^{-c/h}),
$$
with $\delta >0$ arbitrarily small, and $c=c(\delta)>0$. Hence, after the change of variable $t\mapsto h^\frac23 t$, we find,
$$
|\beta^\pm_{L,R}(z)[\varphi[| + |\beta^\pm_{R,R}(z)[\varphi]|= \ord(h^{-\frac56}).
$$
The other estimates follow along the same lines.
\end{proof}

\begin{proposition}\sl 
\label{H0}
$$
R_2^\pm (z) \left(  {\mathcal F}_1(I_L)\cap {\mathcal F}_1(I_R^\pm)\right)\,  \subset \, h^{-2/3}\left( {\mathcal F}_1(I_L)\cap{\mathcal F}_2(I_R^\pm)\right).
$$
\end{proposition}
\begin{proof} Let $\varphi\in {\mathcal F}_1(I_L)\cap {\mathcal F}_1(I_R^\pm)$. By Proposition \ref{propResolv}{(iii)}, om $I_L$ we have,
$$
R_2^\pm (z)\varphi \in K_{2,L}(z)\left(  {\mathcal F}_1(I_L)\right) + \beta_L^\pm (z)[\varphi]h^\frac16  {\mathcal F}_2(I_L),
$$
and therefore, using Propositions \ref{actK} and Remark \ref{rem52},
$$
R_2^\pm (z)\varphi \in h^{-\frac23} {\mathcal F}_1(I_L) + h^{-\frac56+\frac16 } {\mathcal F}_2(I_L) \subset h^{-\frac23} {\mathcal F}_1(I_L).
$$
In the same way, by Proposition \ref{propResolv}{(iv)}, om $I_R^\pm$ we have,
$$
R_2^\pm (z)\varphi \in K_{2,R}^\pm(z)\left(  {\mathcal F}_1(I_R^\pm)\right) + \beta_R^\pm (z)[\varphi]h^\frac16  {\mathcal F}_2(I_R^\pm),
$$
and thus,
$$
R_2^\pm (z)\varphi \in h^{-\frac23}{\mathcal F}_2(I_R^\pm) +h^{-\frac56+\frac16 }   {\mathcal F}_2(I_R^\pm) = h^{-\frac23}{\mathcal F}_2(I_R^\pm).
$$
\end{proof}

As an immediate consequence of this proposition, if we set,
\be
{\mathcal H}_0^\pm:= {\mathcal F}_1(I_L) \cap {\mathcal F}_2(I_R^\pm),
\ee
we have,
\begin{corollary}\sl 
\label{H1}
$$
R_2^\pm (z) W^* \varphi_0 \in h^{-2/3}{\mathcal H}_0^\pm\cap C^\infty.
$$
\end{corollary}

Finally, setting
$$
M_\pm (z):= h^2 R_2^\pm (z)W^* R_1(z)W,
$$
we have,

\begin{proposition}\sl 
\label{prop6.3}
$$
M_\pm (z)\left( {\mathcal H}_0^\pm \cap C^\infty \right)\,\, \subset \,\, h^{1/3}{\mathcal H}_0^\pm \cap C^\infty.
$$
In particular, for any $\ell \geq 1$,
$$
M_\pm (z)^\ell R_2^\pm (z) W^* \varphi_0 \in h^{(\ell -2)/3}{\mathcal H}_0^\pm.
$$
\end{proposition}
\begin{proof} By the same procedure as in the proof of Proposition \ref{H0}, we see,
\be
\label{actR1}
R_1(z)\left( {\mathcal H}_0^\pm\right) \subset h^{-1}{\mathcal F}_1(I_L)  \cap \left[ h^{-\frac23}{\mathcal F}_2(I_R^\pm) + h^{-1}{\mathcal F}_1(I_R^\pm)\right],
\ee
and also,
$$
R_2^\pm(z) \left( {\mathcal F}_1(I_L) \cap h^{\frac13}{\mathcal F}_2(I_R^\pm)\right) \subset h^{-\frac23}{\mathcal H}_0^\pm.
$$
Since, with our definitions, we have,
$$
\begin{aligned}
h^{-1}{\mathcal F}_1(I_L)  \cap & \left[ h^{-\frac23}{\mathcal F}_2(I_R^\pm) + h^{-1}{\mathcal F}_1(I_R^\pm)\right]\\
& = \left[ h^{-1}{\mathcal F}_1(I_L)  \cap h^{-\frac23}{\mathcal F}_2(I_R^\pm)\right] + \left[ h^{-1}{\mathcal F}_1(I_L)  \cap h^{-1}{\mathcal F}_1(I_R^\pm)\right],
\end{aligned}
$$
we deduce,
$$
R_2^\pm(z)W^*R_1(z) \left( {\mathcal H}_0^\pm\right)\subset h^{-\frac53}{\mathcal H}_0^\pm+h^{-1}R_2^\pm (z)\left({\mathcal F}_1(I_L) \cap {\mathcal F}_1(I_R^\pm) \right),
$$
that is, by Proposition \ref{H0},
$$
R_2^\pm(z)W^*R_1(z) \left( {\mathcal H}_0^\pm\right)\subset h^{-\frac53}{\mathcal H}_0^\pm.
$$
Since in addition $W({\mathcal H}_0^\pm \cap C^\infty)\subset {\mathcal H}_0^\pm\cap C^\infty$, and $R_1(z)$, $R_2^\pm (z)$ preserve the regularity at 0, the result follows.
\end{proof}

\section{Estimates on $r_2(t,\phi, h)$}
\label{section7}

We first show,
\begin{lemma}\sl
\label{lemma8.1}
For any $u\in {\mathcal H}_0^\pm$, one has,
\be
\la u, W^*\varphi_0\ra_{L^2(I_L)} ={\mathcal O}(1);  
\ee
\be
\la u, W^*\varphi_0\ra_{L^2( I_R^\pm)} ={\mathcal O}(h^{1/3}).  
\ee
\end{lemma}
\begin{proof} Since $W^*\varphi_0\in {\mathcal F}_1(I_L) \cap {\mathcal F}_1(I_R^\pm) $, the first estimate is immediate, while for the second one, thanks to the exponantial localization near 0 of $\varphi_0\left|_{I_R^\pm}\right.$, we can write,
$$
\la u, W^*\varphi_0\ra_{L^2(I_R^\pm)}=\ord(1)\int_0^\delta \frac{e^{-ct^{3/2}/h}}{\sqrt t} dt+\ord (e^{-c/h}),
$$
(with $\delta >0$ sufficiently small and $c=c(\delta)>0$), and the result follows.
\end{proof}

Then, we have,
\begin{proposition}\sl 
\label{prop8.2}
One has,
$$
r_2(t,\phi, h)={\mathcal O}(h \la ht\ra^{-\infty}).
$$
\end{proposition}
\begin{proof} We must prove that, for any $k\geq 0$, we have $r_2(t,\phi, h)={\mathcal O}(h \la ht\ra^{-k})$. By Corollary \ref{prop5.3}, we already know that there exists a constant $C>0$ such that $|T_\ell (z)|\leq C^\ell h^{-1+(\ell -1)/6}$ uniformly with respect to $h$ small enough. Therefore, for any $L_0\geq 1$,
$$
r_2(t,\phi, h) = \frac{h^2}{2i\pi}\sum_{\ell = 2}^{L_0}\int_{\gamma_-} \frac{e^{-itz} g(\Re z)}{(\lambda_0-z)^2}T_\ell (z) dz +{\mathcal O}(h^{(L_0-1)/6}).
$$
In particular,
$$
r_2(t,\phi, h) = \frac{h^2}{2i\pi}\sum_{\ell = 2}^{6}\int_{\gamma_-} \frac{e^{-itz} g(\Re z)}{(\lambda_0-z)^2}T_\ell (z) dz +{\mathcal O}(h).
$$
Moreover, using Proposition \ref{prop6.3} and Lemma \ref{lemma8.1}, for any $\ell \in \{2,\dots,6\}$, we have,
$$
 T_\ell (z) = {\mathcal O}(h^{(\ell -2)/3}),
$$
and thus,
$$
 \frac{h^2}{2i\pi}\sum_{\ell = 2}^{6}\int_{\gamma_-} \frac{e^{-itz} g(\Re z)}{(\lambda_0-z)^2}T_\ell (z) dz={\mathcal O}(h).
 $$
so that the result for $k=0$ follows. The result for $k\geq 1$ is obtained  by  using that $e^{-itz}=(1+ht)^{-k}(1+ih\partial_z)^k(e^{-itz})$ and by making $k$ integrations by parts. Each derivative $h\partial_z$ that falls on $g(\Re z)(\lambda_0 -z)^{-2}$, doesn't make us lose anything in the estimate. If instead it falls down on $T_\ell (z)$, we need the following,
 \begin{lemma}\sl
 \label{lemma8.3}
 For any $k, \ell \geq 1$, one has,
\be
\label{dzTell}
 h^k\partial_z^k T_\ell (z)={\mathcal O}(h^{ \frac{\ell -2}3}).
 \ee
 Moreover, for any $k, \ell \geq 1$, there exists a constant $C_k$ such that, for all $\ell \geq 1$,
  \be
 \label{dzM}
 \Vert h^k\partial_z^k\left(  M_\pm (z)^\ell\right) \Vert \leq C_k^\ell h^{\frac{\ell}6}.
\ee

 \end{lemma}
 \begin{proof}
 Going back to the construction of the functions $u_{j,L}^\pm(z,x)$ and $u_{j,R}^\pm(z,x)$ (see \cite[Appendix]{FMW1}), we start by observing that they all are of the form $(\partial_x\xi(x,z))^{-1/2}f_z(h^{-2/3}\xi (x,z))$, where $x\mapsto \xi (x,z)$ is a global analytic change of variable that depends analytically on $z$, and $f_z$ is solution to a Volterra problem of the type,
\be
\label{volterra}
  f_z =F + K_z f_z,
\ee
 with $z\mapsto K_z$ holomorphic, and the norm of $K_z$ (and of all its holomorphic derivatives with respect to $z$) is small as $h$ tends to 0 (here, $K_z$ acts on a space continuous functions with some specific growth at infinity depending on the choice of $F$). In addition, the function $F$ appearing in \eqref{volterra} is always taken in the set $\{ {\rm Ai}, {\rm Bi}, \check{\rm Ai}, \check{\rm Bi}\}$. It results that $z\mapsto f_z$ is holomorphic, too, and that, for all $k,\ell$, $\partial_z^k \partial_x^\ell f_z$ growths at most as $\sum_{m=0}^\ell |F^{(\ell)}|$ at infinity. 
 
 Then, considering the function $u_z(x):= f_z(h^{-2/3}\xi (x,z))$, we deduce,
 $$
 \partial_z^k u_z(x) ={\mathcal O}(h^{-2k/3}   \sum_{m=0}^\ell |F^{(\ell)} (h^{-2/3}\xi (x))|).
 $$
 Now, because of the behavior at infinity of the Airy functions, and of the possible choices of the function $F$,  we see that,
 $$
 \sum_{m=0}^\ell |F^{(\ell)} (t)|)={\mathcal O}(\la t\ra ^{\ell /2} F_0(t)),
 $$
 where $F_0$ reflects the behavior of $F$ at infinity, that is, $F_0(t) = \la t\ra^{-1/4} e^{\pm \frac23 |t|^{\frac32}}$ if $F$ has an exponential behavior, and $F_0(t) = \la t\ra^{-1/4}$ if $F$ oscillates at infinity. Therefore, we obtain,
 $$
  \sum_{m=0}^\ell |F^{(\ell)} (h^{-2/3}\xi (x))| ={\mathcal O} (h^{-k/3})F_0(h^{-2/3}\xi (x)),
 $$
 and thus,
\be
\partial_z^k u_z(x) ={\mathcal O}(h^{-k})   F_0(h^{-2/3}\xi (x))).
\ee
In particular, for $j=1,2$, $S\in\{ L,R\}$, and any $k\geq 0$, the function $h^k\partial_z^k u_{j,S}^\pm$ has the same behavior (both semiclassical and at infinity) as the function $u_{j,S}^\pm$ itself. 

As a consequence, considering the operator $h^k\partial_z^k \left(M_\pm(z) \right)^\ell$, we see that it is a sum of $\ell^k$ products of $\ell$ factors, each one of them being of the same type as $M_\pm (z)$, and \eqref{dzM} follows.

For the same reasons, we also have,
$$
h^k\partial_z^k \left(M_\pm(z) \right)^\ell \left({\mathcal H}_0^\pm\right) \,\, \subset\,\, h^{\ell /3}{\mathcal H}_0^\pm,
$$
and
$$
h^k\partial_z^kR_2^\pm (z)W^*\varphi_0 \in {\mathcal H}_0^\pm,
$$
so that  \eqref{dzTell} follows, too.
 \end{proof}
Using Lemma \ref{lemma8.3} and making integrations by parts in the expression of $r_2$ given in \eqref{r0r1}, Proposition \ref{prop8.2} follows.
\end{proof}

\section{Estimates on $r_1(t,\phi, h)$}

 Concerning $r_1(t,\phi, h)$, the same arguments of the previous section can be applied, but they lead to an estimate in ${\mathcal O}(h^{2/3}\la ht\ra^{-\infty})$ only. Let us prove that actually, we have,
 \begin{proposition}\sl 
\label{estr1}
One has,
$$
r_1(t,\phi, h)={\mathcal O}(h \la ht\ra^{-\infty}).
$$
\end{proposition}
\begin{proof}
 In view of Proposition \ref{prop6.3} and Lemma \ref{lemma8.1}, we can write,
 $$
 T_1(z)=\la (M_+(z)R_2^+(z)- M_-(z)R_2^-(z))W^*\varphi_0, W^*\varphi_0\ra_{L^2(I_L)} +{\mathcal O}(1),
 $$
that is, by Proposition \ref{propResolv},
 \be
 \label{T1decomp}
 \begin{aligned}
 T_1(z)= & h^2 \la (K_{2,L}+B_L^+)W^*R_1WR_2^+W^*\varphi_0, W^*\varphi_0\ra_{L^2(I_L)} \\
& -h^2 \la (K_{2,L}+B_L^-)W^*R_1WR_2^-W^*\varphi_0, W^*\varphi_0\ra_{L^2(I_L)} +{\mathcal O}(1),
 \end{aligned}
\ee
 where we have omited the dependence in $z$ of the various operators, and where we have set,
\be
\label{defBpm}
 B_L^\pm \varphi:= \beta_L^\pm (\varphi)u_{2,L}^-.
\ee
Let us first prove,
\begin{lemma}\sl For all $z\in \gamma_-$, one has,
$$
 \begin{aligned}
R_1(z)WR_2^\pm & (z)W^*\varphi_0 \\
& \in h^{-5/3}{\mathcal F}_1 (I_L)\, \cap \,\left( h^{-4/3}{\mathcal F}_2(I_R^\pm) + h^{-5/3}{\mathcal F}_1(I_R^\pm ) \right).
 \end{aligned}
$$
\end{lemma}
\begin{proof}
By Proposition \ref{H0}, we already know that $\psi_\pm:=WR_2^\pm (z)W^*\varphi_0$ is in $h^{-\frac23}{\mathcal H}_0^\pm$. Then, the result directly follows from \eqref{actR1}.
\end{proof}

We deduce from the previous lemma and from Lemma \ref{alphabeta} that we have,
$$
\beta_L^\pm (W^*R_1\psi_\pm) ={\mathcal O} (h^{-15/6}),
$$
and thus,
$$
B_L^\pm W^*R_1\psi_\pm \in h^{-7/3}{\mathcal F}_2(I_L).
$$
As a consequence, by an elementary computation we obtain (using also \eqref{phi2in}),
$$
h^2\la B_L^\pm W^*R_1\psi_\pm , W^*\varphi_0\ra_{L^2(I_L)} ={\mathcal O}(h^{2-\frac73})\int_0^\delta \frac{e^{-cx^{3/2}/h}}{\sqrt x} dx +{\mathcal O}(e^{-\delta'/h}),
$$
(where $\delta$, $\delta'$ and $c$ are positive constants), and thus,
\be
h^2\la B_L^\pm W^*R_1\psi_\pm , W^*\varphi_0\ra_{L^2(I_L)} ={\mathcal O}(1).
\ee
Therefore, going back to \eqref{T1decomp}, we deduce,
$$
 \begin{aligned}
 T_1(z)= & h^2 \la K_{2,L}W^*R_1WR_2^+W^*\varphi_0, W^*\varphi_0\ra_{L^2(I_L)} \\
& -h^2 \la K_{2,L}W^*R_1WR_2^-W^*\varphi_0, W^*\varphi_0\ra_{L^2(I_L)} +{\mathcal O}(1),
 \end{aligned}
$$
that is,
\be
 \label{T1decompbis}
 \begin{aligned}
 T_1(z)= & h^2 \la K_{2,L}W^*(K_{1,L}+ A_L)WR_2^+W^*\varphi_0, W^*\varphi_0\ra_{L^2(I_L)} \\
& -h^2 \la K_{2,L}W^*(K_{1,L}+ A_L)WR_2^-W^*\varphi_0, W^*\varphi_0\ra_{L^2(I_L)} +{\mathcal O}(1),
 \end{aligned}
\ee
where this time we have set,
$$
A_L:= A_{LL} + A_{LR}, 
$$
with
$$
A_{LL}(\varphi):= \alpha_{LL}(\varphi_L)u_{1,L}^-\quad ; \quad A_{LR}(\varphi):= \alpha_{LR}(\varphi_R)u_{1,L}^-.
$$
In particular, setting,
$$
B_R^\pm \varphi := \beta_R^\pm (\varphi)u_{2,R}^\mp,
$$
we can rewrite \eqref{T1decompbis} as,
$$
 \begin{aligned}
 T_1(z)= & h^2 \la K_{2,L}W^*(K_{1,L}+ A_{LL})W(K_{2,L}+B_L^+)W^*\varphi_0, W^*\varphi_0\ra_{L^2(I_L)} \\
& -h^2 \la K_{2,L}W^*(K_{1,L}+ A_{LL})W(K_{2,L}+B_L^-)W^*\varphi_0, W^*\varphi_0\ra_{L^2(I_L)} \\
 &+ h^2\la K_{2,L}W^*A_{LR}W(K_{2,R}+B_R^+)W^*\varphi_0, W^*\varphi_0\ra_{L^2(I_L)} \\
 &-h^2\la K_{2,L}W^*A_{LR}W(K_{2,R}+B_R^-)W^*\varphi_0, W^*\varphi_0\ra_{L^2(I_L)} +{\mathcal O}(1),
 \end{aligned}
$$
that is, after having eliminated the terms that cancel,
\be
 \label{T1decompter}
 \begin{aligned}
 T_1(z)= & h^2 \la K_{2,L}W^*(K_{1,L}+ A_{LL})W(B_L^+-B_L^-)W^*\varphi_0, W^*\varphi_0\ra_{L^2(I_L)} \\
 &+ h^2\la K_{2,L}W^*A_{LR}W(B_R^+-B_R^-)W^*\varphi_0, W^*\varphi_0\ra_{L^2(I_L)}  +{\mathcal O}(1).
 \end{aligned}
\ee
Using again Lemma  \ref{alphabeta} and \eqref{ujLRin}, we find,
\be
\label{BLphiin}
WB_L^\pm W^*\varphi_0\in h^{-2/3}{\mathcal F}_2(I_L)\quad ; \quad  WB_R^\pm W^*\varphi_0\in h^{-2/3}{\mathcal F}_2(I_R^\pm),
\ee
and thus also,
$$
 \begin{aligned}
&W^*A_{LL}WB_L^\pm W^*\varphi_0\in h^{-4/3}{\mathcal F}_1(I_L);\\
 &W^*A_{LR}WB_R^\pm W^*\varphi_0\in h^{-4/3}{\mathcal F}_1(I_L).
 \end{aligned}
$$
Therefore, by Proposition \ref{actK},
$$
\begin{aligned}
& K_{2,L}W^*A_{LL}WB_L^\pm W^*\varphi_0\in h^{-2}{\mathcal F}_1(I_L);\\
&   K_{2,L}W^*A_{LR}WB_R^\pm W^*\varphi_0\in h^{-2}{\mathcal F}_1(I_L).
\end{aligned}
$$
As a consequence,  we obtain (with $S=L,R$),
$$
h^2\la K_{2,L}W^*A_{LS}WB_S^\pm W^*\varphi_0,W^*\varphi_0\ra_{L^2(I_L)}={\mathcal O}(1),
$$
and \eqref{T1decompter} reduces to,
\be
 \label{T1decomp4}
 T_1(z)=  h^2 \la K_{2,L}W^*K_{1,L}W(B_L^+-B_L^-)W^*\varphi_0, W^*\varphi_0\ra_{L^2(I_L)} +{\mathcal O}(1).
\ee
Using \eqref{BLphiin} and, once more, Proposition \ref{actK}, we obtain,
$$
K_{2,L}WK_{1,L}W^*B_L^\pm W^*\varphi_0 \in K_{2,L}\left( h^{-4/3}{\mathcal F}_1(I_L)\right)\subset h^{-2}{\mathcal F}_1(I_L),
$$
so that the same computations finally give us,
\be
T_1(z)=  {\mathcal O}(1).
\ee
Hence, $r_1(t,\varphi, h)={\mathcal O}(h)$ and Proposition \ref{estr1} follows by the same arguments as in the previous section.
\end{proof}

\section{Estimates on $r_0(t,\phi, h)$}

By definition we have,
$$
\begin{aligned}
T_0(z) = &  \la (R_2^+(z)-R_2^-(z))W^*\varphi_0, W^*\varphi_0\ra_{L^2(I_L)}\\
& + \la R_2^+(z)W^*\varphi_0, W^*\varphi_0\ra_{L^2(I_R^+)}
  -\la R_2^-(z)W^*\varphi_0, W^*\varphi_0\ra_{L^2(I_R^-)},
\end{aligned}
$$
and thus, by Proposition \ref{propResolv},
\be
\label{T=T1+T2+T3}
T_0(z) =T_{0,1}(z)+T_{0,2}(z)+T_{0,3}(z)
\ee
with  (using the same notations \eqref{defBpm} as in the previous section),
\be
\begin{aligned}
T_{0,1}(z) = &  \la (B_L^+(z)-B_L^-(z))W^*\varphi_0, W^*\varphi_0\ra_{L^2(I_L)};\\
T_{0,2}(z) =&  \la B_R^+(z)W^*\varphi_0, W^*\varphi_0\ra_{L^2(I_R^+)}
 - \la B_R^-(z)W^*\varphi_0, W^*\varphi_0\ra_{L^2(I_R^-)};\\
 T_{0,3}(z) =&  \la K_{2,R}^+(z)W^*\varphi_0, W^*\varphi_0\ra_{L^2(I_R^+)}
 - \la K_{2,R}^-(z)W^*\varphi_0, W^*\varphi_0\ra_{L^2(I_R^-)}.
\end{aligned}
\ee

We first show,
\begin{proposition}\sl
\label{prop9.1}
Setting $\rho := h^{-\frac23}z$ and $\mu_0:=h^{-\frac23}\lambda_0$,   one has,
$$
\begin{aligned}
& \la u_{2,L}^-(z),  W^*\varphi_0\ra_{L^2(I_L)} = 4a_0(0)c_0 h^{\frac12}A^{-}(\rho) + \ord(h^{\frac56});\\
& \la u_{2,R}^-(z),  W^*\varphi_0\ra_{L^2(I_R^{+})} = \sqrt{2}a_0(0)c_0 h^{\frac12}e^{i\frac{\pi}4}\left(A^{+}(\rho)-iB^{+}(\rho)\right)+ \ord(h^{\frac56});\\
& \la u_{2,R}^+(z),  W^*\varphi_0\ra_{L^2(I_R^{-})} = 2\sqrt{2}a_0(0)c_0 h^{\frac12}e^{i\frac{\pi}4}\left(A^{+}(\rho)+iB^{+}(\rho)\right)+ \ord(h^{\frac56}),
\end{aligned}
$$
with,
$$
\begin{aligned}
& A^{\pm}(\rho):= \tau_1^{-\frac16}\tau_2^{-\frac16}\int_{\R_{\pm}} \check\ai \left(\tau_2^{\frac13}(y+\frac{\rho}{\tau_2}\right)\ai \left(\tau_1^{\frac13}(y-\frac{\mu_0}{\tau_1}\right)dy;\\
& B^{+}(\rho):= \tau_1^{-\frac16}\tau_2^{-\frac16}\int_0^{+\infty} \check\bi \left(\tau_2^{\frac13}(y+\frac{\rho}{\tau_2}\right)\ai \left(\tau_1^{\frac13}(y-\frac{\mu_0}{\tau_1}\right)dy.
\end{aligned}
$$
Moreover, the same formulas hold if $W^*\varphi_0$ is substituted by $\overline{ W^*\varphi_0}$.
\end{proposition}
\begin{proof} The proof is similar to that of \cite[Proposition 5.3]{FMW1}. In practical, we cut the integral on $I_{L}$ into $\int_{-\infty}^{-\lambda h^\frac23} +\int_{-\lambda h^\frac23}^0$, and that on $I_{R}^{+}$ into $\int_{0}^{\lambda h^\frac23} +\int_{I_{R}^{+}\cap \{\re x \geq \lambda h^\frac23\}}^0$, where $\lambda = C\ln |h|$, with $C>0$ a large enough constant. Then, we use the exponential decay of $u_{2,L}^-(z)$ away from 0 on $\R_{-}$, and that of $W^{*}\varphi_{0}$ away from 0 on $I_{R}^{+}$, in order to estimate the integrals on $(-\infty ,{-\lambda h^\frac23}]$ and on $I_{R}^{+}\cap \{\re x \geq \lambda h^\frac23\}$, and finally, near 0 we replace $u_{2,L}^-(z)$, $u_{2,R}^-(z)$ and $\varphi_0$ by their approximations in terms of Airy functions (see \cite[Appendix 2]{FMW1},
\be
\label{u2LR-}
\begin{aligned}
& u_{2,L}^-(z,x)=2(\xi_2')^{-\frac12}\check\ai (h^{-\frac23}\xi_2)+\ord (h);\\
& u_{2,R}^-(z,x)=\frac1{\sqrt 2}e^{i\frac{\pi}4}(\xi_2')^{-\frac12}\left(\check\ai (h^{-\frac23}\xi_2)-i\check\bi (h^{-\frac23}\xi_2)\right)+\ord (h);\\
& u_{2,R}^+(z,x)={\sqrt 2}e^{i\frac{\pi}4}(\xi_2')^{-\frac12}\left(\check\ai (h^{-\frac23}\xi_2)+i\check\bi (h^{-\frac23}\xi_2)\right)+\ord (h);\\
& \varphi_0(x)=2c_{0}h^{-\frac16}(\xi_1')^{-\frac12}\ai (h^{-\frac23}\xi_1) + \ord (h),
\end{aligned}
\ee

where $\xi_1=\xi_1(x)$ and $\xi_2(z,x)$ satisfy (see \cite[Section7]{FMW1}),
\be
\begin{aligned}
& h^{-\frac23}\xi_1(h^\frac23 y)  =\tau_1^\frac13 \left(y-\frac{\mu_0}{\tau_1}\right)+\ord(h^\frac23);\\
& h^{-\frac23}\xi_2(z, h^\frac23 y)  =\tau_2^\frac13 \left(y+\frac{\rho}{\tau_1}\right)+\ord(h^\frac23).
\end{aligned}
\ee
\end{proof}
Then, 
we prove,

\begin{proposition}\sl
Still with $\rho:= h^{-\frac23}z$, one has,
\be
\label{expT01}
T_{0,1}(z)=8i\pi h^{-\frac13}c_0^2a_0(0)^2\left[A^{+}(\rho)+A^{-}(\rho)\right]A^{-}(\rho) +\ord (1);
\ee
\be
\label{expT02}
T_{0,2}(z)= 4i\pi h^{-\frac13}c_0^2a_0(0)^2 \left[ 2A^{+}(\rho)A^{-}(\rho)+A^{+}(\rho)^{2}-B^{+}(\rho)^{2}\right] +\ord (1);
\ee
\be
\label{expT03}
T_{0,3}(z)=4i\pi h^{-\frac13}c_0^2a_0(0)^2\left[ A^{+}(\rho)^{2}+B^{+}(\rho)^{2}\right]+\ord (1).
\ee

\end{proposition}
\begin{proof}
Setting,
\be
\varphi_1 := W^*\varphi_0,
\ee
by definition we have,
$$
T_{0,1} = \left[ \beta_{L,L}^+(\varphi_1)-\beta_{L,L}^-(\varphi_1)+\beta_{L,R}^+(\varphi_1)-\beta_{L,R}^-(\varphi_1)\right]\la u_{2,L}^-,\varphi_1\ra.
$$
Then, using Remark \ref{rem52} and proceeding as in the proof of Proposition \ref{prop9.1}, we see,
$$
\begin{aligned}
  \beta_{L,L}^+(\varphi_1)-\beta_{L,L}^-(\varphi_1)& =i\frac{\pi}2 h^{-\frac43}(1+\ord(h^{\frac13}))\int_{-\infty}^0u_{2,L}^-(z,t)\varphi_1(t)dt\\
  &= 2i\pi a_0(0)c_0A^-(\rho)h^{-\frac56} + \ord(h^{-\frac12});
\end{aligned}
$$
$$
\begin{aligned}
  \beta_{L,R}^+(\varphi_1)& =\frac{\pi}{\sqrt 2}a_0(0)c_0e^{i\frac{\pi}4}h^{-\frac56}\frac2{\sqrt 2}e^{i\frac{\pi}4}(A^+(\rho)-iB^+(\rho))  + \ord(h^{-\frac12}) \\
  &= i\pi a_0(0)c_0(A^+(\rho)-iB^+(\rho))h^{-\frac56} + \ord(h^{-\frac12});\\
  \beta_{L,R}^-(\varphi_1)  &=- i\pi a_0(0)c_0(A^+(\rho)+iB^+(\rho))h^{-\frac56} + \ord(h^{-\frac12}).
  \end{aligned}
$$
Hence,
$$
\begin{aligned}
 \beta_{L,L}^+(\varphi_1)-\beta_{L,L}^-(\varphi_1)+&\beta_{L,R}^+(\varphi_1)-\beta_{L,R}^-(\varphi_1)\\
 &=2i\pi a_0(0)c_0(A^-(\rho)+A^+(\rho))h^{-\frac56} + \ord(h^{-\frac12}),
  \end{aligned}
$$
which, together with Proposition \ref{prop9.1}, gives \eqref{expT01}.

In the same way, we have,
$$
\begin{aligned}
& \beta_{R,L}^+(\varphi_1) = \frac{4\pi}{\sqrt 2}e^{i\frac{\pi}4}a_0(0)c_0A^-(\rho)h^{-\frac56}+ \ord(h^{-\frac12});\\
&\beta_{R,R}^+(\varphi_1)=\frac{2\pi}{\sqrt 2}e^{i\frac{\pi}4}a_0(0)c_0\left( A^+(\rho)-iB^+(\rho) \right)h^{-\frac56}+ \ord(h^{-\frac12});\\
& \beta_{R,L}^-(\varphi_1) =-\frac{2\pi}{\sqrt 2}e^{i\frac{\pi}4}a_0(0)c_0A^-(\rho) h^{-\frac56}+ \ord(h^{-\frac12}) ;\\
&\beta_{R,R}^-(\varphi_1)=-\frac{\pi}{\sqrt 2}e^{i\frac{\pi}4}a_0(0)c_0\left( A^+(\rho) +iB^+(\rho)\right)h^{-\frac56}+ \ord(h^{-\frac12}) . ,
\end{aligned}
$$
and thus, by Proposition  \ref{prop9.1},
$$
\begin{aligned}
\la B_R^+\varphi_1,\varphi_1\ra_{L^2(I_R^+)}=2i \pi a_0(0)^2c_0^2 &\left( 2A^-(\rho) +A^+(\rho)-iB^+(\rho) \right)\\
&  \hskip 1.5cm\times \left( A^+(\rho)-iB^+(\rho) \right)h^{-\frac13}+ \ord(1); \\
\la B_R^-\varphi_1,\varphi_1\ra_{L^2(I_R^-)}=-2i \pi a_0(0)^2c_0^2 &\left( 2A^-(\rho) +A^+(\rho)+iB^+(\rho) \right)\\
&  \hskip 1.5cm\times \left( A^+(\rho)+iB^+(\rho) \right)h^{-\frac13}+ \ord(1).
\end{aligned}
$$
Then, \eqref{expT02} follows by a straightforward computation.

Finally, concerning $T_{0,3}$, using the exponential decay of $\varphi_1$ on $I_R^\pm$ away from 0, for any $\delta >0$ small enough we can write,
$$
T_{0,3} = \la K_{2,R}^+\varphi_1 -K_{2,R}^-\varphi_1, \varphi_1\ra_{L^2(0,\delta)} + \ord( e^{-c/h}),
$$
with $c =c(\delta) >0$ constant. Now, we see on \eqref{eq1R+}-\eqref{eq1R-} that, for $x\in [0,\delta]$, we have (dropping the dependance in $z$),
$$
\begin{aligned}
&K_{2,R}^+\varphi_1(x) -K_{2,R}^-\varphi_1(x)\\
& \hskip 1cm=\frac1{h^2 \W[u_{2,R}^-,u_{2,R}^+]} \int_0^x \left(u_{2,R}^-(x)u_{2,R}^+(t) +u_{2,R}^-(t),u_{2,R}^+(x)\right) \varphi_1(t)dt\\
&  \hskip 1.5 cm+\frac1{h^2 \W[u_{2,R}^-,u_{2,R}^+]} \int_x^\delta \left(u_{2,R}^-(t)u_{2,R}^+(x) +u_{2,R}^-(x),u_{2,R}^+(t)\right) \varphi_1(t)dt\\
&  \hskip 1.5 cm+\ord(e^{-c'/h}),
\end{aligned}
$$
with $c'=c'(\delta) >0$ constant, and therefore, 
$$
\begin{aligned}
&K_{2,R}^+\varphi_1(x) -K_{2,R}^-\varphi_1(x)\\
& \hskip 1cm=\frac1{h^2 \W[u_{2,R}^-,u_{2,R}^+]} \int_0^\delta \left(u_{2,R}^-(x)u_{2,R}^+(t) +u_{2,R}^-(t),u_{2,R}^+(x)\right) \varphi_1(t)dt\\
&  \hskip 1.5 cm+\ord(e^{-c'/h}).
\end{aligned}
$$
Using also that $\W[u_{2,R}^-,u_{2,R}^+] = \frac2{\pi}h^{-\frac23}+\ord(h^{-\frac13}$) (see \cite[Appendix A.2]{FMW1}, and the fact that $\varphi_1(x)$ is real for $x$ real, we conclude,
$$
\begin{aligned}
T_{0,3} =& \left( \pi h^{-\frac43}+\ord(h^{-1})\right)\la u_{2,R}^-, \varphi_1\ra_{I_R^+} \la u_{2,R}^+ , \varphi_1\ra_{I_R^-} + \ord( e^{-c''/h}),
\end{aligned}
$$
(still with $c''>0$ constant). Hence, \eqref{expT03} follows by using Proposition \ref{prop9.1}.
\end{proof}

We conclude from the previous proposition and \eqref{T=T1+T2+T3} that we have,
\be
T_{0}(z) = 8i\pi h^{-\frac13}c_0^2a_0(0)^2 f(h^{-\frac23}z) +\ord (1),
\ee
with,
$$
f(\rho):=\left(A^{-}(\rho)+A^{+}(\rho)\right)^{2}.
$$

Then, writing,
$$
f(h^{-\frac23}z)=f(h^{-\frac23}\lambda_{0})+\ord (h^{-\frac23})\, (z-\lambda_{0}),
$$
and going back to \eqref{r0r1}, we obtain,
$$
r_{0}(t,\phi,h) =4h^{2-\frac13}c_0^2a_0(0)^2 f(h^{-\frac23}\lambda_0)\int_{\gamma_-} \frac{e^{-itz}g(\re z)}{(\lambda_0-z)^2}dz +\ord (h),
$$
and then, by arguments similar to those of Section \ref{section7},
\be
r_{0}(t,\phi,h) =4h^{2-\frac13}c_0^2a_0(0)^2 f(h^{-\frac23}\lambda_0)\int_{\gamma_-} \frac{e^{-itz}g(\re z)}{(\lambda_0-z)^2}dz +\ord (h\la ht\ra^{-\infty}).
\ee
Finally, using \eqref{defg} and making the change of variable $z\mapsto \lambda_0+hz$, we find,
$$
r_{0}(t,\phi,h) =4h^{\frac23}c_0^2a_0(0)^2e^{-it\lambda_0} f(h^{-\frac23}\lambda_0)\int_{\gamma_0} \frac{e^{-ithz}g_0(\re z)}{z^2}dz +\ord (h\la ht\ra^{-\infty}),
$$
and \eqref{estq0} is proved with,
\be
\label{defintA0}
\begin{aligned}
 A_0(s)= \tau_1^{-\frac16}\tau_2^{-\frac16}\int_{-\infty}^{+\infty} \check\ai \left(\tau_2^{\frac13}(y+\frac{s}{\tau_2}\right)\ai \left(\tau_1^{\frac13}(y-\frac{s}{\tau_1}\right)dy.
\end{aligned}
\ee
The fact that one also has,
\be
\label{formuleA0}
A_0(s) =  \tau_1^{-\frac16}\tau_2^{-\frac16}(\tau_1+\tau_2)^{-\frac13} \check\ai \left(\left(\frac{\tau_1+\tau_2}{\tau_1\tau_2} \right)^{\frac23}s \right)
\ee
is proved in Appendix 2.
\qed

\section{Estimate on $b(\phi, h)$}

Let $\Psi_0\in L^2(I_L\cup I_R^+$ be the resonant state associated with $\rho_0$, that is the solution to $H\Psi_0 =\rho_0\Psi_0$ normalized in such a way that,
\be
\label{normres}
\la \Psi_0^\theta, \Psi_0^{-\theta}\ra =1,
\ee
where $ \Psi_0^{\pm\theta}:=U_{\pm \theta}\Psi_0$. Then, by definition of $b(\phi, h)$, we have,
$$
b(\phi,h)=\la \phi_\theta, \Psi_0^{-\theta}\ra\la \Psi_0^\theta, \phi_{-\theta}\ra.
$$
According to \cite[Remark 7.1]{FMW1}, we have,
$$
\Psi_0 = \left( \begin{array}{c}
C_1 u_{1,L}^-\downharpoonright_{z=\rho_0}+\ord(h^\frac23)\\
\ord(h^\frac13)
\end{array}
\right)\quad \mbox{on }I_L;
$$
$$
\Psi_0 = \left( \begin{array}{c}
(-1)^k C_1 u_{1,R}^-\downharpoonright_{z=\rho_0}+\ord(h^\frac23)\\
\ord(h^\frac13)
\end{array}
\right)
\quad \mbox{on }I_R^+,
$$
where  the integer $k\geq 0$ is such that $\sin \frac{{\mathcal A}(\rho_0)}{h} =(-1)^k +\ord (h^\frac23)$, and the coefficient $C_1=C_1(h)$ is such that \eqref{normres} is verified.
We also have,
$$
\phi =  \left( \begin{array}{c}
 c_0u_{1,L}^-\downharpoonright_{z=\lambda_0} \\
0
\end{array}
\right)\quad \mbox{on }I_L;
$$
$$
\phi = \left( \begin{array}{c}
(-1)^k c_0 u_{1,R}^-\downharpoonright_{z=\lambda_0}\\
0
\end{array}
\right)
\quad \mbox{on }I_R^+.
$$
Moreover, $\rho_0 -\lambda_0 =\ord (h^\frac43)$ (see \cite[Section2]{FMW1}), and thus, since $\partial_zu_{1,L}^-$ is $\ord (h^{-1})$ locally uniformly in $x$ (and exponentially decays at infinity), we get,
$$
\Psi_0 = \left( \begin{array}{c}
C_1 u_{1,L}^-\downharpoonright_{z=\lambda_0}\\
0
\end{array}
\right)+\ord(h^\frac13)\quad \mbox{on }I_L;
$$
$$
\Psi_0 = \left( \begin{array}{c}
(-1)^k C_1 u_{1,R}^-\downharpoonright_{z=\lambda_0}\\
0
\end{array}
\right)+\ord(h^\frac13)
\quad \mbox{on }I_R^+.
$$
As a consequence, we necessarily have $C_1=\pm c_0(1+\ord (h^\frac13))$, and \eqref{estbh} follows.\qed

\section{Appendix 1: proof of Proposition \ref{actK}}

We only consider the case of $I_L$ (the one of $I_R^\pm$ being similar).

Let $v\in {\mathcal F}_1(I_L)$. We have,
\be
\label{K1Lv}
K_{1,L} v(x) = \ord (h^{-\frac43})\int_{I_L}\left( u_{1,L}^-(x)u_{1,L}^+(t){\bf 1}_{x<t} +u_{1,L}^-(t)u_{1,L}^+(x){\bf 1}_{x>t}\right)v(t) dt
\ee
We fix $\delta >0$ arbitrarily small, and we first suppose that \underline{$x\leq x^*-\delta$}. In the integral of \eqref{K1Lv}, we decompose $I_L$ into three parts: 
$$
I_L = (-\infty, x^*-\frac{\delta}2]\cup [x^*-\frac{\delta}2, -\delta]\cup [-\delta, 0].
$$
$\bullet$ On $(-\infty, x^*-\frac{\delta}2]$: There, we have $u_{1,L}^-(x)u_{1,L}^+(t){\bf 1}_{x<t} +u_{1,L}^-(t)u_{1,L}^+(x){\bf 1}_{x>t}=\ord (h^{\frac13}e^{-c_0|x-t|/h})$ and $v(t) = \ord (e^{-c_1|t|/h})$, with $c_0,c_1 >0$ constants. Using that $|x-t|+|t|\geq \frac12 (|x|+|t|)$, we obtain,
$$
h^{-\frac43}\int_{-\infty}^{x^*-\frac{\delta}2}\left( u_{1,L}^-(x)u_{1,L}^+(t){\bf 1}_{x<t} +u_{1,L}^-(t)u_{1,L}^+(x){\bf 1}_{x>t}\right)v(t) dt =\ord (h^{-1} e^{-c_2|x|/h}),
$$
with $c_2:= \frac12 \min (c_0,c_1)$.

$\bullet$ On $(x^*-\frac{\delta}2, -\delta]$: There, we have $u_{1,L}^-(x)u_{1,L}^+(t){\bf 1}_{x<t} +u_{1,L}^-(t)u_{1,L}^+(x){\bf 1}_{x>t}=u_{1,L}^-(x)u_{1,L}^+(t)=\ord (h^{\frac13}e^{-c_0|x|/h}m_*(t))$ and $v(t) = \ord (m_*(t))$, with $c_0 >0$ constant, and  we obtain,
$$
h^{-\frac43}\int_{x^*-\frac{\delta}2}^{-\delta}\left( u_{1,L}^-(x)u_{1,L}^+(t){\bf 1}_{x<t} +u_{1,L}^-(t)u_{1,L}^+(x){\bf 1}_{x>t}\right)v(t) dt =\ord (h^{-1} e^{-c_0|x|/h}).
$$

$\bullet$ On $( -\delta, 0]$: There, we have $u_{1,L}^-(x)u_{1,L}^+(t){\bf 1}_{x<t} +u_{1,L}^-(t)u_{1,L}^+(x){\bf 1}_{x>t}=u_{1,L}^-(x)u_{1,L}^+(t)=\ord (h^{\frac13}e^{-c_0|x|/h}m_0(t))$ and $v(t) = \ord ( \,m_0(t))$, with $c_0 >0$ constant, and  we obtain,
$$
h^{-\frac43}\int_{-\delta}^0\left( u_{1,L}^-(x)u_{1,L}^+(t){\bf 1}_{x<t} +u_{1,L}^-(t)u_{1,L}^+(x){\bf 1}_{x>t}\right)v(t) dt =\ord ( h^{-1}e^{-c_0|x|/2h}).
$$

Thus, on $(-\infty, x^*-\delta]$, we have,
$$
K_{1,L} v(x) = \ord (e^{-c_3|x|/h}),
$$
with $c_3>0$ constant.

Suppose now that \underline{$x^*-\delta \leq x\leq x^*$}. This time we divide $I_L$ into,
$$
I_L=(-\infty, x^*-2\delta]\cup [x^*-2\delta, x^*]\cup [x^*, -\delta]\cup [-\delta, 0].
$$

$\bullet$ On $(-\infty, x^*-2\delta]$: There, we have $u_{1,L}^-(x)u_{1,L}^+(t){\bf 1}_{x<t} +u_{1,L}^-(t)u_{1,L}^+(x){\bf 1}_{x>t}=u_{1,L}^-(t)u_{1,L}^+(x)=\ord (e^{-c_0/h}m_*(x))$ and $v(t) = \ord (e^{-c_1|t|/h})$, with $c_0,c_1 >0$ constants, and we obtain,
$$
h^{-\frac43}\int_{-\infty}^{x^*-2\delta}\left( u_{1,L}^-(x)u_{1,L}^+(t){\bf 1}_{x<t} +u_{1,L}^-(t)u_{1,L}^+(x){\bf 1}_{x>t}\right)v(t) dt =\ord (m_*(x) e^{-c_0/2h}).
$$

$\bullet$ On $[x^*-2\delta, x^*]$: There, we have $u_{1,L}^-(x)u_{1,L}^+(t){\bf 1}_{x<t} +u_{1,L}^-(t)u_{1,L}^+(x){\bf 1}_{x>t}=\ord (h^{\frac13}m_*(x)m_*(t)e^{-c_0\left|(x^*-x)^{\frac32}-(x^*-t)^\frac32\right|/h}$ and $v(t) = \ord (m_*(t)e^{-c_1(x^*-t)^\frac32/h})$, with $c_0,c_1 >0$ constants, and making the change of variable $t\mapsto x^*-t$, and using the notation $\widetilde x:= x^*-x$, we obtain,
$$
\begin{aligned}
h^{-\frac43}\int_{x^*-2\delta}^{x^*}& \left( u_{1,L}^-(x)u_{1,L}^+(t){\bf 1}_{x<t} +u_{1,L}^-(t)u_{1,L}^+(x){\bf 1}_{x>t}\right)v(t) dt \\
&=\ord(h^{-1}m_*(x))\int_0^{2\delta} \frac{e^{-\left( c_0|t^\frac32 -\widetilde x^\frac32|+c_1t^\frac32 \right)/h}}{\sqrt t} dt.
\end{aligned}
$$
Thus, using the fact that $|t^\frac32 -\widetilde x^\frac32|+t^\frac32\geq \frac12( \widetilde x^\frac32 + t^\frac32)$, this gives us,
$$
\begin{aligned}
h^{-\frac43}\int_{x^*-2\delta}^{x^*}& \left( u_{1,L}^-(x)u_{1,L}^+(t){\bf 1}_{x<t} +u_{1,L}^-(t)u_{1,L}^+(x){\bf 1}_{x>t}\right)v(t) dt \\
&=\ord(h^{-1}m_*(x)e^{-c_2\widetilde x^\frac32})\int_0^{2\delta} \frac{e^{-c_2t^\frac32 /h}}{\sqrt t} dt=\ord(h^{-\frac23}m_*(x)e^{-c_2\widetilde x^\frac32}),
\end{aligned}
$$
with $c_2 =\frac12 \min (c_0,c_1)$.

$\bullet$ On $[x^*, -\delta]$: There, we have $u_{1,L}^-(x)u_{1,L}^+(t){\bf 1}_{x<t} +u_{1,L}^-(t)u_{1,L}^+(x){\bf 1}_{x>t}=u_{1,L}^-(x)u_{1,L}^+(t)=\ord (h^{\frac13}m_*(x)m_*(t)e^{-c_0(x^*-x)^{\frac32}/h}$ and $v(t) = \ord (m_*(t))$, with $c_0 >0$ constant, and  we obtain,
$$
\begin{aligned}
h^{-\frac43}\int_{x^*}^{-\delta} &\left( u_{1,L}^-(x)u_{1,L}^+(t){\bf 1}_{x<t} +u_{1,L}^-(t)u_{1,L}^+(x){\bf 1}_{x>t}\right)v(t) dt \\
&=\ord(h^{-1}m_*(x)e^{-c_0(x^*-x)^{\frac32}/h}).
\end{aligned}
$$

$\bullet$ On $[ -\delta, 0]$: There, we have $u_{1,L}^-(x)u_{1,L}^+(t){\bf 1}_{x<t} +u_{1,L}^-(t)u_{1,L}^+(x){\bf 1}_{x>t}=u_{1,L}^-(x)u_{1,L}^+(t)=\ord (h^{\frac13}m_*(x)m_0(t)e^{-c_0(x^*-x)^{\frac32}/h}$ and $v(t) = \ord ( \,m_0(t))$, with $c_0 >0$ constant, and  we obtain,
$$
\begin{aligned}
h^{-\frac43}\int_{-\delta}^0 &\left( u_{1,L}^-(x)u_{1,L}^+(t){\bf 1}_{x<t} +u_{1,L}^-(t)u_{1,L}^+(x){\bf 1}_{x>t}\right)v(t) dt \\
&=\ord( \,h^{-1}m_*(x)e^{-c_0(x^*-x)^{\frac32}/h}).
\end{aligned}
$$
Thus, on $(x^*-\delta, x^*]$, we have,
$$
K_{1,L} v(x) = \ord ( \,h^{-1}m_*(x)e^{-c_3(x^*-x)^{\frac32}/h}),
$$
with $c_3>0$ constant.

Then we consider the case \underline{$x^*\leq x\leq -\delta$}. We divide $I_L$ into,
$$
I_L=(-\infty, x^*-\delta]\cup [x^*-\delta, x^*]\cup [x^*, 0].
$$
Arguing as before, we find,
$$
\begin{aligned}
h^{-\frac43}\int_{-\infty}^{x^*-\delta}&\left( u_{1,L}^-(x)u_{1,L}^+(t){\bf 1}_{x<t} +u_{1,L}^-(t)u_{1,L}^+(x){\bf 1}_{x>t}\right)v(t) dt\\
& =\ord (m_*(x)e^{-c_0/h})
\end{aligned}
$$

$$
\begin{aligned}
h^{-\frac43}\int_{x^*-\delta}^{x^*}&\left( u_{1,L}^-(x)u_{1,L}^+(t){\bf 1}_{x<t} +u_{1,L}^-(t)u_{1,L}^+(x){\bf 1}_{x>t}\right)v(t) dt\\
& =\ord( m_*(x)h^{-\frac23})
\end{aligned}
$$

$$
\begin{aligned}
h^{-\frac43}\int_{x^*}^0&\left( u_{1,L}^-(x)u_{1,L}^+(t){\bf 1}_{x<t} +u_{1,L}^-(t)u_{1,L}^+(x){\bf 1}_{x>t}\right)v(t) dt\\
& =\ord ( m_*(x) h^{-1}),
\end{aligned}
$$
and thus, on $[ x^*, -\delta]$, we have,
$$
K_{1,L} v(x) = \ord ( \,h^{-1}m_*(x)).
$$

 Finally, in the case \underline{$-\delta \leq x\leq 0$}, dividingagain $I_L$ into,
 $$
I_L=(-\infty, x^*-\delta]\cup [x^*-\delta, x^*]\cup [x^*, 0],
 $$
 we find in the same way,
 $$
 K_{1,L} v(x) = \ord (\, h^{-1} m_0(x)).
 $$
 We also see that the same estimates hold for the derivatives $(hD_x)^kK_{1,L} v(x) $, and thus
 we have proved,
 $$
K_{1,L}\left( {\mathcal F}_1(I_L)\right) \subset   h^{-1}{\mathcal F}_1(I_L).
 $$
 
 Concerning $K_{1,L}\left( {\mathcal F}_2(I_L)\right)$, that is, if $v\in {\mathcal F}_2(I_L)$, the same decompositions as before give exponentially small terms only, multiplied by $m_*(x)m_0(x)$, except those for $t$ close to 0. For these last ones, the previous arguments permit us to estimate them by,
 $$
 \ord(h^{-1})m_*(x)m_0(x)\alpha (x)\int_0^\delta \frac{e^{-t^{\frac32}/h}}{\sqrt t}dt = \ord( h^{-\frac23})m_*(x)m_0(x)\alpha (x),
 $$
 where we have used the notation $\alpha (x):=e^{-c(x^*-x)^{\frac32}/h})$ (with $c>0$ constant) when $x\leq x^*$, and $\alpha (x) :=1$ when $x\geq x^*$. This proves that,
 $$
 K_{1,L}\left( {\mathcal F}_2(I_L)\right) \subset  h^{-2/3}{\mathcal F}_1(I_L).
 $$
 
 The estimate on $ K_{2,L}\left( {\mathcal F}_1(I_L)\right)$ follows essentially in the same way, except for the behaviour near $x=x^*$. Take $v\in {\mathcal F}_1(I_L)$, and first consider $K_{2,L}(v)(x)$ for $x^*-\delta \leq x\leq x^*$. One has,
 $$
 \begin{aligned}
 K_{2,L}(v)(x)=& \ord(h^{-1})\int _{x^*-\delta}^{x^*} e^{-c_0(|t-x|+|t-x^*|^{\frac32})/h}m_*(t)dt\\
 &+\ord(h^{-1})\int _{x^*}^{x^*
 +\delta} e^{-c_0|t-x|/h}m_*(t)dt +\ord (e^{-c_0/h}),
\end{aligned}
 $$
 with $c_0>0$.
 Making the change of variable $t\mapsto x^*-t$, and setting $\widetilde x := x^*-x$, we obtain,
 $$
  \begin{aligned}
 K_{2,L}(v)(x)=& \ord(h^{-1})\int _{0}^{\delta} e^{-c_0(|\widetilde x-t|+t^{\frac32})/h}t^{-\frac14}dt\\
 &+\ord(h^{-1})\int _{ -\delta}^0 e^{-c_0(\widetilde x - t)/h}|t|^{-\frac14}dt +\ord (e^{-c_0/h}),
\end{aligned}
 $$
 and thus, using that $\left| t^{\frac32} -\widetilde x^\frac32\right| \leq \frac12 |\widetilde x -t|$ for $t, \widetilde x\in [0, \delta]$, $\delta$ small enough, we deduce,
 $$
  \begin{aligned}
 K_{2,L}(v)(x)=& \ord(h^{-1})e^{-c_0\widetilde x^\frac32 /h}\int _0^\delta e^{-c_0|\widetilde x-t|/2h}t^{-\frac14}dt\\
 &+\ord(h^{-1})\int _0^\delta e^{-c_0(\widetilde x + t)/h}t^{-\frac14}dt +\ord (e^{-c_0/h}),
\end{aligned}
 $$
 Making the change of variable $t\mapsto ht$, this gives us,
  $$
  \begin{aligned}
 K_{2,L}(v)(x)=& \ord(h^{-\frac14})e^{-c_0\widetilde x^\frac32 /h}\int _0^{\delta /h} e^{-c_0(h^{-1}\widetilde x-t)/2}t^{-\frac14}dt\\
 &+\ord(h^{-\frac14})\int _0^{\delta/h} e^{-c_0(h^{-1}\widetilde x + t)}t^{-\frac14}dt +\ord (e^{-c_0/h}),
\end{aligned}
 $$
 and, cutting the first integral into $\int_0^1 + \int_1^{\widetilde x /h} + \int_{\widetilde x/h}^{\delta /h}$ in the case $\widetilde x \geq h$, and into $\int_0^{1}+\int_{1}^{\delta /h}$ in the case $0\leq \widetilde x \leq h$, we obtain,
 $$
 K_{2,L}(v)(x)=\ord(h^{-\frac14})e^{-c_0\widetilde x^\frac32 /h}+ \ord(h^{-\frac14})e^{-c_0\widetilde x /h}
 $$
  and therefore (since $0\leq \widetilde x^\frac32 \leq \widetilde x$), for $x^*-\delta \leq x\leq x^*$,
\be
\label{estK2LF1.1}
 K_{2,L}(v)(x)=\ord(h^{-\frac14})e^{-c_0\widetilde x^\frac32 /h}=\ord(h^{-\frac14})e^{-c_0(x^*-x)^\frac32 /h}.
\ee
When $x^*\leq x\leq x^*+\delta$, the same computations lead to 
\be
\label{estK2LF1.2}
 K_{2,L}(v)(x)=\ord(h^{-\frac14}).
\ee
 Then, when $x^*+\delta\leq x\leq -\delta$, we can write,
 $$
 K_{2,L}(v)(x)=\ord( h^{-1})e^{-|x|^\frac32/h}\int_0^\delta e^{t^\frac32/h}t^{-\frac12}dt +\ord (1),
$$
and the change of variable $t\mapsto (ht)^{\frac23}$ gives us,
\be
\label{estK2LF1.3}
 K_{2,L}(v)(x)=\ord( h^{-\frac23})e^{-|x|^\frac32/h}\int_0^{\delta^\frac32/h} \frac{e^{t}}{t^\frac23}dt +\ord (1)=\ord( h^{-\frac23}).
\ee
Finally, for $-\delta\leq x\leq 0$, the same kind of computations lead to,
\be
\label{estK2LF1.4}
K_{2,L}(v)(x)=\ord( h^{-\frac23})m_0(x).
\ee
Since in addition $h^{-\frac14}\leq h^{-\frac23}m_*(x)$ on $[x^*-\delta, x^*+\delta]$, the required result on $ K_{2,L}\left( {\mathcal F}_1(I_L)\right)$ follows from \eqref{estK2LF1.1}-\eqref{estK2LF1.4}

 The estimate on $ K_{2,L}\left( {\mathcal F}_2(I_L)\right)$ follows along the same lines, together with the results on $I_R^\pm$.

\section{Appendix 2: proof of \eqref{formuleA0}}
\label{A2}

For any tempered function $f=f(x)$ on $\R$, we denote by $\hat f$ its Fourier transform defined by,
$$
\hat f(\xi):= \int_{-\infty}^{+\infty} e^{-ix\xi} f(x) dx,
$$
and, for $\alpha >0$ constant and $x\in\R$, we set,
$$
f_\alpha (x):= \ai (\alpha x).
$$
By definition we have $\hat\ai (\xi) =e^{i\xi^3/3}$, and thus $\hat f_\alpha (\xi) = \alpha^{-1}e^{i\alpha^{-3}\xi^3/3}$.

Then, we see on \eqref{defintA0} that we have,
$$
A_0(s)=g(x):= \tau_1^{-\frac16}\tau_2^{-\frac16} (f_\alpha * f_\beta)(x)
$$
with,
$$
\alpha:= \tau_1^{\frac13}\quad ;\quad \beta:= \tau_2^{\frac13}\quad ;\quad x:= -\left( \frac{\tau_1+\tau_2}{\tau_1\tau_2}\right) s,
$$
and where $*$ stands for the standard convolution of functions. As a consequence, using that $\hat{(f_\alpha *f_\beta)} =\hat f_\alpha \hat f_\beta$, we obtain, 
$$
\begin{aligned}
\hat g(\xi) & =\tau_1^{-\frac16}\tau_2^{-\frac16} \hat f_\alpha (\xi)\hat f_\beta (\xi) =\tau_1^{-\frac16}\tau_2^{-\frac16}\alpha^{-1}\beta^{-1}e^{i(\alpha^{-3}+\beta^{-3})\xi^3/3}\\
&= \tau_1^{-\frac16}\tau_2^{-\frac16}\gamma\alpha^{-1}\beta^{-1}\hat f_\gamma(\xi),
\end{aligned}
$$
with,
$$
\gamma:= (\alpha^{-3}+\beta^{-3})^{-\frac13}=\left( \frac{\tau_1\tau_2}{\tau_1+\tau_2}\right)^\frac13.
$$
Hence,
$$
g=\tau_1^{-\frac16}\tau_2^{-\frac16}\gamma\alpha^{-1}\beta^{-1}f_\gamma = \tau_1^{-\frac16}\tau_2^{-\frac16}(\tau_1+\tau_2)^{-\frac13}f_\gamma,
$$
and \eqref{formuleA0} follows.



{}

\end{document}